\documentclass[journal]{IEEEtran}
\usepackage{amsmath,amsfonts}
\usepackage{array}
\usepackage[caption=false,font=scriptsize,labelfont=rm,textfont=rm]{subfig}
\usepackage{textcomp}
\usepackage{stfloats}
\usepackage{threeparttable}
\usepackage{url}
\usepackage{verbatim}
\usepackage{graphicx}
\def\BibTeX{{\rm B\kern-.05em{\sc i\kern-.025em b}\kern-.08em
		T\kern-.1667em\lower.7ex\hbox{E}\kern-.125emX}}
\usepackage{balance}

\usepackage{multirow}
\usepackage{cite}
\usepackage{cases}
\usepackage{amssymb}
\usepackage{algorithm}
\usepackage{algpseudocode}
\usepackage{bm}
\usepackage{color,xcolor}

\usepackage{float}    
\usepackage{subfig}   
\usepackage{overpic}  

\usepackage{amsthm}

\newtheorem{Lemma}{Lemma}
\newtheorem{mypro}{Proposition}
\allowdisplaybreaks[3]
\begin{document}

	\title{Outage-Constrained Environment Division Multiple Access (EDMA) for Pinching-Antenna  Systems}
	
	\author{Weihao Mao,~\IEEEmembership{Student Member,~IEEE},~Yang Lu,~\IEEEmembership{Senior Member,~IEEE},~Yanqing Xu,~\IEEEmembership{Senior Member,~IEEE}, \\ and Zhiguo Ding,~\IEEEmembership{Fellow,~IEEE}
		\thanks{Weihao Mao and Yang Lu are with the School of Computer Science and Technology, Beijing Jiaotong University, Beijing 100044, China (e-mail: weihaomao,yanglu@bjtu.edu.cn).}
        \thanks{Yanqing Xu  is with the School of Science and Engineering, The Chinese University of Hong Kong, Shenzhen, China. (e-mail: xuyanqing@cuhk.edu.cn)}
        \thanks{Zhiguo Ding is with the School of Electrical and Electronic Engineering (EEE), Nanyang Technological University, Singapore 639798 (Zhiguo.ding@ntu.edu.sg).}
	}

	\maketitle

	
	\begin{abstract}
    Environment division multiple access (EDMA) has emerged as a promising multiple access paradigm, which mitigates inter-user interference by dynamically adjusting pinching antenna (PA) positions to the underlying propagation environment. This paper investigates a multi-user PA-enabled EDMA framework that accounts for probabilistic line-of-sight (LoS) blockages, random non-LoS (NLoS) scattering, and practical in-waveguide attenuation. With the aim of maximizing the total information rate, we formulate a joint PA deployment and power allocation problem subject to statistical rate outage constraints, the transmit power budget, and the feasible deployment region of PAs. We first consider a canonical two-user two-PA scenario and derive closed-form expressions for the outage probabilities, followed by a low-complexity projected gradient descent (PGD)-based algorithm to address the reformulated problem. Then, we extend our design to the general multi-user multi-PA scenario and derive tractable approximations for the outage probabilities by assuming interfering links to be NLoS and applying the Chernoff bounding technique, where a successive convex approximation (SCA)-based algorithm is proposed to handle the resulting non-convex problem. Simulations validate the superiority of the proposed PA-enabled EDMA design and the effectiveness of the proposed algorithms. Specifically, both the PGD-based algorithm and the SCA-based algorithm achieve near-optimal performance in comparison with the exhaustive search. Furthermore, the PA-enabled EDMA design yields significant performance gains over both PA-enabled and conventional time division multiple access designs. 
    \end{abstract}

    \begin{IEEEkeywords} 
    Environment division multiple access (EDMA), pinching antenna (PA), probabilistic LoS blockage, outage-constrained optimization 
	\end{IEEEkeywords}
	
	\IEEEpeerreviewmaketitle
	
	\setlength{\parindent}{1em}
	\section{Introduction}

    Sixth generation (6G) mobile communication networks impose stringent requirements for ultra-reliable and low-latency communications \cite{bc_6g}. Nevertheless, the explosive proliferation of intelligent devices and the escalating demands for high data rates severely conflict with the scarce spectrum resources \cite{bc_spectrum}. To address this bottleneck, emerging flexible-antenna technologies have been proposed to improve spectral efficiency by dynamically reconfiguring channel conditions \cite{bc_flex}. Specifically, fluid antennas (FAs) \cite{bc_fa} and movable antennas (MAs) \cite{bc_ma} enable spatial channel customization via software-controlled radio-frequency (RF) switches and mechanical drivers, respectively. However, positional adjustments of FAs and MAs are typically restricted to the wavelength scale \cite{bc_ma2}, which limits their performance in complicated environments \cite{bc_pa1}. To overcome this confinement, pinching antennas (PAs) have recently emerged as a novel architecture \cite{bc_pa2}. By leveraging dedicated dielectric waveguides, PAs allow antennas to be dynamically activated at any point within their deployed waveguides, which supports large-scale antenna deployment \cite{bc_pa3}. 

    Compared with FAs and MAs, PAs can mitigate the large-scale path loss \cite{bc_pa_pl} and sustain line-of-sight (LoS) communication links \cite{bc_pa_los}, which has attracted significant attention from academia. Initial studies focused on single-waveguide PA architectures. For instance, \cite{rw_pa_sing_w1} derived the rate region for PA systems, demonstrating their superiority over conventional fixed-position antennas. Furthermore, \cite{rw_pa_sing_w2} integrated non-orthogonal multiple access (NOMA) techniques into PA systems and investigated the rate maximization problem in a two-user scenario. To further exploit the spatial degrees of freedom, subsequent research considered multi-waveguide systems. Specifically, \cite{rw_pa_mul_w1} jointly optimized the pinching beamforming and power allocation to maximize the sum rate, while \cite{rw_pa_mul_w2} investigated the joint design of antenna activation and resource allocation in a NOMA-assisted PA system. Beyond pure communication systems, recent works have explored the integration of diverse network functionalities into PA systems. For example, \cite{rw_pa_mul_f1} proposed a fine-tuning-based optimization method to maximize the communication rate subject to the radar signal-to-noise ratio (SNR) requirement in PA-enabled integrated sensing and communication (ISAC) systems. Additionally, \cite{rw_pa_mul_f2} incorporated mobile edge computing (MEC) into PA systems and considered a delay minimization problem by jointly optimizing the PA positions and resource allocation. 

    The existing literature \cite{rw_pa_sing_w1, rw_pa_sing_w2, rw_pa_mul_w1, rw_pa_mul_w2, rw_pa_mul_f1, rw_pa_mul_f2} relies on the assumption that LoS links always exist between PAs and users, which, however, rarely holds in practical obstacle-dense scenarios \cite{bc_pro_pa1, bc_pro_los1, bc_pro_los2}. In reality, environmental blockages are not necessarily detrimental and can be employed by PAs to mitigate inter-user interference via intentionally suppressing the interference links into non-LoS (NLoS) states \cite{bc_pro_pa2}. This insight motivates the emergence of \emph{environment division multiple access} (EDMA), a novel framework that leverages LoS blockages and severe path loss as spatial resources to facilitate quasi-interference-free multiple access and information transmission \cite{bc_pro_pa3}. Existing studies on PA-enabled EDMA can be generally classified into two categories according to the prior knowledge of LoS links: deterministic LoS channel-based designs and probabilistic LoS channel-based designs. For the former, \cite{rw_deter_pa1} and \cite{rw_deter_pa2} investigated the PA deployment strategies in environments characterized by cylindrical obstacles and rectangular cuboid obstacles, respectively. 
    For the latter, \cite{los_block} analyzed and maximized the ergodic rate of PA-enabled EDMA systems under the probabilistic LoS channel model. The results from \cite{rw_deter_pa1,rw_deter_pa2,los_block} demonstrate the advantages of EDMA over conventional multiple access techniques such as time division multiple access (TDMA). However, existing studies fail to guarantee transmission reliability since they neglect the adverse impact of link blockages on the communication outage performance.


    \emph{So far, the outage-constrained design for PA-enabled EDMA systems remains unexplored in the open literature.} Despite its importance for ultra-reliable communications, realizing such a design is highly non-trivial due to the difficulties in deriving a tractable form for the rate outage probability. \emph{To fill this gap,} this paper investigates the joint design of PA deployment and power allocation for outage-constrained PA-enabled EDMA systems, where both probabilistic LoS blockages and random NLoS scattering are taken into account. Furthermore, the practical in-waveguide attenuation is also considered to guarantee physical-layer fidelity. The main contributions of this paper are summarized as follows:

    \begin{itemize}
        \item 
        We propose a multi-user PA-enabled EDMA framework, where distributed PAs are coordinated by a base station (BS) to facilitate concurrent transmissions. Moreover, the impact of the probabilistic LoS blockages, random NLoS scattering, and in-waveguide attenuation is considered in the system model. We formulate a total information rate maximization problem by jointly optimizing the PA deployment, power allocation, and information rate, subject to statistical rate outage constraints, the transmit power budget, and the feasible deployment region of PAs.
        \item 
        To obtain analytical insights, we first consider a two-user two-PA scenario, where tractable closed-form expressions for the outage probabilities are derived. Building upon this, we reformulate the optimization problem by decoupling variables leveraging optimality analysis, followed by a projected gradient descent (PGD)-based algorithm to handle the reformulated problem. To further reduce the computational complexity, we derive the closed-form optimal solution for each projection subproblem based on the Karush-Kuhn-Tucker (KKT) conditions. 
        \item
        To extend our design to the general multi-user multi-PA scenario, we derive tractable approximations for the outage probabilities by simplifying the interfering links to be NLoS and applying the Chernoff bounding technique. On this basis, we introduce a series of auxiliary variables to decouple the optimization variables, followed by a successive convex approximation (SCA)-based algorithm to tackle the reformulated problem. 
        \item 
        Extensive simulation results validate the superiority of the proposed PA-enabled EDMA framework and the effectiveness of the proposed PGD-based and SCA-based algorithms. The derived analytical and approximated outage probabilities tightly match the Monte Carlo results, rigorously justifying our analytical assumptions. Moreover, the proposed PGD-based and SCA-based algorithms exhibit near-optimal performance in comparison with the exhaustive search. Furthermore, the PA-enabled EDMA design yields substantial performance gains over both PA-enabled TDMA and conventional TDMA designs.
    \end{itemize}
    
   {\it Notations:} In this paper, the symbols $x$, ${\bf x}$, ${\bf X}$, and ${\cal X}$ represent a scalar, a vector, a matrix, and a set, respectively. $|| \cdot ||$ denotes the Euclidean norm for a vector. ${\bf a} \sim \mathcal{CN}({\bm \mu}, {\bm \Sigma})$ denotes that ${\bf a}$ is a circularly symmetric complex Gaussian random vector with mean ${\bm \mu}$ and covariance matrix ${\bm \Sigma}$. $\mathbb{P}(\cdot)$ and $\mathbb{E}(\cdot)$ denote the probability operator and the statistical expectation operator, respectively.

	\section{System Model and Problem Formulation}
	\begin{figure}
		\begin{center}
			\centerline{\includegraphics[ width=.46\textwidth]{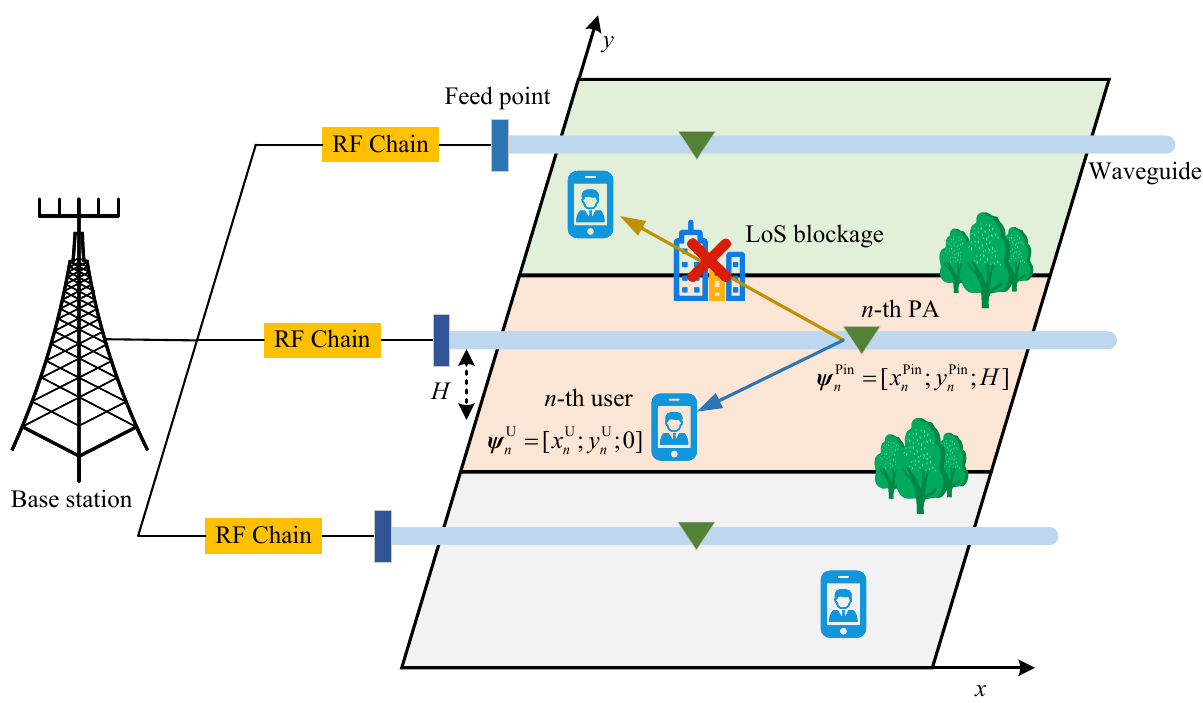}}
			\caption{Illustration of a PA-enabled EDMA system.}
			\label{sys}
            \vspace{-2em}
		\end{center}
	\end{figure}

In this paper, we consider a PA-enabled EDMA system as illustrated in Fig. \ref{sys}, where a BS coordinates $N$ PAs to deliver information services to $N$ users. Each PA is deployed on a dedicated dielectric waveguide connected to an independent RF chain through its corresponding feed point, which enables flexible and independent signal transmission and reception across different waveguides. By exploiting the parallelly deployed waveguides, the entire service region is partitioned into $N$ non-overlapping sub-regions, where the $n$-th sub-region is denoted by ${\cal  A}_n$ ($n\in\{1,2,\dots,N\}$).  Without loss of generality, we assume that within each sub-region, only one user is scheduled and only one PA is implemented. To further alleviate the signal processing complexity, we adopt a one-to-one association between users and PAs, i.e., each user is served exactly by the PA located in the same sub-region and vice versa. This assumption is motivated by the fact that dense building obstacles can severely block the LoS links between users and PAs located in distinct sub-regions \cite{los_block}. For notational brevity, the user residing in sub-region ${\cal A}_n$ is denoted by ${\rm U}_n$, and its corresponding serving PA is denoted by ${\rm Pin}_n$.

The service region is assumed to be a rectangle, with its length and width denoted by $D_{\rm L}$ and $D_{\rm W}$, respectively. For illustrative purposes, all sub-regions are assumed to have identical sizes, i.e., each sub-region takes the shape of a rectangle with size $D_{\rm L} \times {D_{\rm W}}/{N}$. A three-dimensional (3D) Cartesian coordinate system is adopted to characterize the positions of users and PAs. Under this coordinate system, the coordinates of ${\rm U}_m$ and ${\rm Pin}_n$ are given by ${\bm \psi}_m^{\rm U} \triangleq [x_m^{\rm U};y_m^{\rm U};0]$ and ${\bm \psi}_n^{\rm Pin} \triangleq [x_n^{\rm Pin}; y_n^{\rm Pin}; H]$, respectively, where $H$ denotes the deployment height of dielectric waveguides.

\subsection{Channel Model}

	\begin{figure}
		\begin{center}
			\centerline{\includegraphics[ width=.46\textwidth]{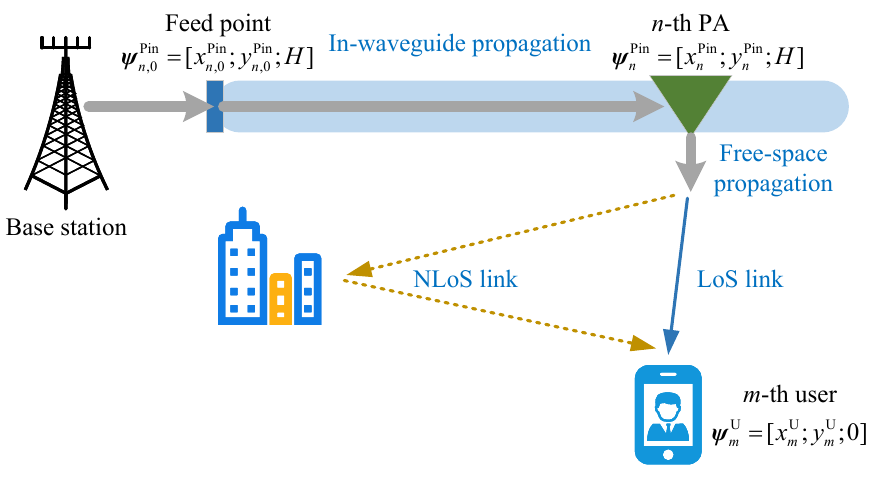}}
			\caption{Illustration of the channel between ${\rm Pin}_n$ and ${\rm U}_m$.}
			\label{pinch_channel}
            \vspace{-2em}
		\end{center}
	\end{figure}

As illustrated in Fig. \ref{pinch_channel}, distinct from conventional wireless communication systems, the signal in PA-enabled systems experiences in-waveguide attenuation and free-space fading. We denote the channel associated with ${\rm Pin}_n$ and ${\rm U}_{ m}$ by $h_{n, m}$, which is expressed as
\begin{flalign}
    h_{n, m} = \underbrace{\frac{ e^{- j \frac{2 \pi}{\lambda} n_{\rm eff} \left| \left| {\bm \psi}_n^{\rm Pin} -{\bm \psi}_{n, 0}^{\rm Pin}    \right| \right|} }{ e^{\alpha \left| \left| {\bm \psi}_n^{\rm Pin} - {\bm \psi}_{n, 0}^{\rm Pin}  \right| \right| } }}_{\rm \text{in-waveguide~attenuation}}
    \times
    \underbrace{\frac{\sqrt{\eta}}{\left| \left| {\bm \psi}_{n}^{\rm Pin} - {\bm \psi}_m^{\rm U} \right| \right| }  g_{n, m}}_{\rm \text{free-space~fading}} ,
\end{flalign}
where $\lambda$ denotes the wavelength of the carrier frequency, $n_{\rm eff}$ denotes the effective refractive index of a dielectric waveguide \cite{neff}, ${\bm \psi}_{n, 0}^{\rm Pin} \triangleq [ x_{n, 0}^{\rm Pin}; y_{n, 0}^{\rm Pin}; H]$ denotes the coordinate of the feed point of the $n$-th waveguide, $\alpha$ denotes the in-waveguide attenuation coefficient determined by the waveguide structure and frequency \cite{in_wave_att}, $\eta \triangleq (\lambda / (4 \pi))^2$ is a frequency-dependent constant \cite{eta_para}, and $g_{n, m}$ denotes the small-scale fading component associated with $h_{n, m}$.



Furthermore, we consider a general and practical scenario where the LoS link between a PA and a user may be subject to blockage. Under a probabilistic LoS channel model \cite{p_LoS_channel}, the small-scale fading component $g_{n, m}$ is given by
\begin{flalign}
    g_{n, m} = \zeta_{n, m} g_{n, m}^{\rm LoS} + (1 - \zeta_{n, m}) \kappa g_{n, m}^{\rm NLoS},
\end{flalign}
where $\zeta_{n, m} \in \{0, 1\}$ is a Bernoulli random variable with $\zeta_{n, m} = 1$ indicating the existence of an LoS link, $0< \kappa < 1$ denotes the additional signal attenuation factor incurred by NLoS propagation, $g_{n, m}^{\rm LoS} \triangleq e^{- j  \frac{2 \pi}{\lambda} ||{\bm \psi}_n^{\rm Pin} - {\bm \psi}_m^{\rm U} || }$ denotes the LoS component of $g_{n, m}$, and $g_{n, m}^{\rm NLoS} \sim \mathcal{CN}(0, 1)$ denotes the Rayleigh fading component corresponding to NLoS propagation. Based on \cite{p_LoS_channel2}, the probability of the LoS link existence is given by
\begin{flalign}
    \mathbb{P} ( \zeta_{n, m} = 1 ) = e^{- \beta || {\bm \psi}_n^{\rm Pin} - {\bm \psi}_m^{\rm U} ||^2 }, \label{los_prob}
\end{flalign}
where $\beta \leq 1$ is a small positive parameter that characterizes the density of blockages in the environment. 

\textit{Remark 1:} Equation \eqref{los_prob} reveals that the LoS availability probability is dominated by the distance between ${\rm Pin}_n$ and ${\rm U}_m$. This property highlights a key advantage of PA-enabled systems, namely, the capability to flexibly realize antenna activation and deployment on a large scale. For instance, ${\rm Pin}_n$ can be deployed in close proximity to its served ${\rm U}_n$ to maintain a high-quality LoS link $h_{n, n}$, while being sufficiently separated from other users (e.g., ${\rm U}_m$) such that the interference link $h_{n, m}$ ($m \neq n$) is prone to blockage, thereby suppressing the interference inflicted on ${\rm U}_m$ by information signals associated with ${\rm U}_n$. This constitutes the core principle of the considered EDMA design, i.e., \emph{leveraging the inherent blockages in the propagation environment to mitigate inter-user interference}.

\subsection{Downlink Communication Model}

The downlink signal received by ${\rm U}_n$ is given by
\begin{flalign}
    y_n^{\rm d} = \underbrace{\sqrt{p_n} h_{n, n} s_n}_{\text{desired signal}} + \underbrace{\sum\nolimits_{m \neq n}^N \sqrt{p_m} h_{m, n} s_m}_{\text{inter-user interference}} + n_n,
\end{flalign}
where $p_n \geq 0$ denotes the transmit power allocated to ${\rm U}_n$, $s_n \sim \mathcal{CN}(0, 1)$ denotes the information symbol intended for ${\rm U}_n$, and $n_n \sim \mathcal{CN}(0, \sigma_n^2)$ denotes the additive white Gaussian noise (AWGN) at ${\rm U}_n$ with $\sigma_n^2$ representing the corresponding noise power. 

Then, the signal-to-interference-plus-noise ratio (SINR) at ${\rm U}_n$ is formulated as 
\begin{flalign}
    {\Gamma}_n = \frac{p_n |h_{n,n}|^2}{ \sum_{m \neq n}^N p_m |h_{m, n}|^2 + \sigma_n^2 }. \label{sinr}
\end{flalign}
It is worth noting that $\Gamma_n$ in equation \eqref{sinr} exhibits inherent randomness, which stems from the probabilistic LoS blockage characterized by the Bernoulli random variable $\zeta_{m, n}$ and the random Rayleigh fading of the NLoS component $g_{m, n}^{\rm NLoS}$.

Due to the probabilistic LoS blockages and random NLoS scattering, the instantaneous channel gains are inherently stochastic and highly dynamic. To ensure the reliability of the considered system, an outage probability constraint is introduced to account for the aforementioned uncertainty, which is given by
\begin{flalign}
    \mathbb{P} \left( \log_2\left( 1 + \Gamma_n \right) \leq R_n  \right) \leq \epsilon,
\end{flalign}
where $R_n$ denotes the target information rate to be optimized for ${\rm U}_n$, and $\epsilon \in (0, 1)$ denotes the maximum tolerable outage probability.

\subsection{Problem Formulation}

In this paper, we aim to maximize the total information rate via the joint optimization of PA deployment $\{x_n^{\rm Pin}\}$, transmit power allocation $\{p_n\}$, and information rate $\{R_n\}$ under the constraints of the statistical rate outage, the total transmit power budget, and the feasible deployment region of PAs. Specifically, the considered joint optimization problem is formulated as 
\begin{subequations}
	\begin{flalign}
		{{\bf P}_0}: & \mathop{\max}  \limits_{\{ x_n^{\rm Pin}, p_n, R_n \} } \sum\nolimits_{n=1}^N R_n   
			  \label{p0a}   \\
		{\rm s.t.} ~ & \mathbb{P} \left( \log_2 \left( 1 + \Gamma_n  \right) \leq R_n \right) \leq \epsilon,  \forall n \in \{1, 2, \ldots, N\}, \label{p0b}
       \\
        & \sum\nolimits_{n = 1}^N p_n \leq P_{\max}, p_n \geq 0, \forall n \in \{1, 2, \ldots, N\},
        \label{p0c} \\
        &  0 \leq x_n^{\rm Pin} \leq D_{\rm L}, \forall n \in \{1, 2, \ldots, N\},  \label{p0d}
	\end{flalign}
\end{subequations}
where $P_{\max}$ denotes the total transmit power budget.

Problem ${\bf P}_0$ is non-convex owing to the intricate coupling among the PA deployment positions $\{ x_n^{\rm Pin} \}$, transmit power allocation $\{p_n\}$, and information rate $\{R_n\}$. Furthermore, no explicit closed-form expression is available for the rate outage probability in constraint \eqref{p0b}, which further exacerbates the difficulty of solving the formulated problem.

\section{Performance Analysis for the Two-User Two-PA Case}

In this section, we investigate a special case of the PA-enabled EDMA system consisting of only two users and two PAs. For this special case, a tractable closed-form expression can be derived for the outage probability in constraint \eqref{p0b}. Then, we decouple the optimization variables within the constraints via rigorous optimality analysis and propose a PGD-based algorithm to address the reformulated problem.

\subsection{Outage Probability Analysis}

Recall that the outage probability in \eqref{p0b} is governed by two sources of randomness, i.e., the binary Bernoulli-distributed LoS blockages and the Rayleigh fading inherent in the NLoS propagation channels. 

We first analyze the LoS blockages characterized by $\zeta_{m, n}$. For the considered two-user two-PA special case, by exhaustively enumerating all possible LoS/NLoS channel states, the outage probability for ${\rm U}_1$ is explicitly expressed as
\begin{flalign}
    & 
    \mathbb{P} \left( \log_2 \left( 1 + \Gamma_1 \right) \leq R_1  \right) = \mathbb{P}\left( \zeta_{1, 1} = 1 \right) \mathbb{P} \left( \zeta_{2, 1} = 1 \right) \nonumber \\
    & \times
    \mathbb{P}\left( \Gamma_1 \leq 2^{R_1} - 1 | \{ \zeta_{1,1},  \zeta_{2, 1} \} = \{1, 1\}  \right) + \mathbb{P} \left( \zeta_{1,1} = 1 \right)  \nonumber  \\
    & \times
    \mathbb{P} \left( \zeta_{2, 1} = 0 \right) \mathbb{P} \left( \Gamma_1 \leq 2^{R_1} - 1 | \{ \zeta_{1, 1}, \zeta_{2, 1} \} = \{1, 0\} \right) \nonumber \\
    & + 
    \mathbb{P}\left( \zeta_{1, 1} = 0 \right) \mathbb{P} \left( \zeta_{2, 1} = 1 \right) \nonumber \\
    & \times
    \mathbb{P} \left( \Gamma_1 \leq 2^{R_1} - 1 | \{ \zeta_{1,1}, \zeta_{2, 1} \} = \{0, 1\}  \right) + \mathbb{P} \left( \zeta_{1, 1} = 0 \right) \nonumber \\
    & \times 
    \mathbb{P} \left( \zeta_{2, 1} = 0 \right) \mathbb{P} \left( \Gamma_1 \leq 2^{R_1} - 1 | \{ \zeta_{1, 1}, \zeta_{2, 1} \} = \{0, 0\} \right).
    \label{2_pa_prob_1}
\end{flalign}
After some algebraic manipulation, \eqref{2_pa_prob_1} can be rewritten as \eqref{2_pa_prob_2}, 
\begin{figure*}
\begin{flalign}
     \mathbb{P} \left( \log_2 \left( 1 + \Gamma_1 \right) \leq R_1  \right) = &e^{-\beta || {\bm \psi}_1^{\rm Pin} - {\bm \psi}_1^{\rm U}||^2 }  e^{ - \beta || {\bm \psi}_2^{\rm Pin} - {\bm \psi}_1^{\rm U}  ||^2 } F_1 \left( p_1, p_2, x_1^{\rm Pin}, x_2^{\rm Pin}, R_1  \right) \nonumber \\
    & + e^{ - \beta || {\bm \psi}_1^{\rm Pin} - {\bm \psi}_1^{\rm U}||^2 } \left( 1 - e^{ -\beta || {\bm \psi}_2^{\rm Pin} - {\bm \psi}_1^{\rm U}  ||^2 }  \right)  \mathbb{P} \left( \frac{ \eta p_1  e^{-2 \alpha ||{\bm \psi}_1^{\rm Pin} - {\bm \psi}_{1, 0}^{\rm Pin} || } }{||{\bm \psi}_1^{\rm Pin} - {\bm \psi}_1^{\rm U}  ||^2} - \left( 2^{R_1} - 1 \right) \sigma_1^2 \leq   z_{2, 1}  \right)  \nonumber \\
    & + \left( 1 - e^{ - \beta || {\bm \psi}_1^{\rm Pin} - {\bm \psi}_1^{\rm U}  ||^2  }  \right) e^{ - \beta || {\bm \psi}_2^{\rm Pin} - {\bm \psi}_1^{\rm U}  ||^2 }  \mathbb{P} \left( z_{1, 1} \leq \left( 2^{R_1} - 1 \right) \left( \frac{ \eta p_2  e^{-2 \alpha ||{\bm \psi}_2^{\rm Pin} - {\bm \psi}_{2, 0}^{\rm Pin} || } }{||{\bm \psi}_2^{\rm Pin} - {\bm \psi}_1^{\rm U}  ||^2}   + \sigma_1^2    \right)  \right) \nonumber \\
    & + \left( 1 - e^{-\beta || {\bm \psi}_1^{\rm Pin} - {\bm \psi}_1^{\rm U}  ||^2 }  \right) \left(  1 -  e^{-\beta || {\bm \psi}_2^{\rm Pin} - {\bm \psi}_1^{\rm U}   ||^2} \right) \mathbb{P} \left( z_{1, 1} \leq z_{2, 1} + \left( 2^{R_1} - 1 \right)\sigma_1^2  \right)
    \label{2_pa_prob_2}
\end{flalign}
\hrule
\end{figure*}
where $F_1(p_1, p_2, x_1^{\rm Pin}, x_2^{\rm Pin}, R_1)$ denotes the indicator function, and is given by
\begin{flalign}
    & F_1 \left( p_1, p_2, x_1^{\rm Pin}, x_2^{\rm Pin}, R_1  \right) \triangleq \nonumber \\
    & \begin{cases}
        1, & \text{if}~ \frac{\eta p_1 e^{ - 2 \alpha \left| \left| {\bm \psi}_1^{\rm Pin} - {\bm \psi}_{1, 0}^{\rm Pin}  \right| \right| } }{\left| \left| {\bm \psi}_{1}^{\rm Pin} - {\bm \psi}_1^{\rm U} \right| \right|^2 }  \leq  \\
        &  ~~~\left( 2^{R_1} - 1   \right)   \left( \frac{ \eta p_2 e^{-2 \alpha ||{\bm \psi}_2^{\rm Pin} - {\bm \psi}_{2, 0}^{\rm Pin} || } }{||{\bm \psi}_2^{\rm Pin} - {\bm \psi}_1^{\rm U}  ||^2} + \sigma_1^2   \right),    \\ 
        0, & \text{Otherwise}.
    \end{cases}
\end{flalign}
Moreover, $z_{1, 1}$ and $z_{2, 1}$ are introduced random variables, which are respectively defined by
\begin{flalign}
    &
    z_{1, 1} \triangleq \frac{ \eta p_1 \kappa^2  }{||{\bm \psi}_1^{\rm Pin} - {\bm \psi}_1^{\rm U}  ||^2 e^{2 \alpha ||{\bm \psi}_1^{\rm Pin} - {\bm \psi}_{1, 0}^{\rm Pin} || }}  \left|  g_{1, 1}^{\rm NLoS}    \right|^2,
    \\
    &
    z_{2, 1} \triangleq \frac{ \eta p_2 \kappa^2 \left( 2^{R_1} - 1 \right) }{||{\bm \psi}_2^{\rm Pin} - {\bm \psi}_1^{\rm U}  ||^2 e^{2 \alpha ||{\bm \psi}_2^{\rm Pin} - {\bm \psi}_{2, 0}^{\rm Pin} || }}  \left| g_{2, 1}^{\rm NLoS}   \right|^2. 
\end{flalign}

Next, we address the random NLoS Rayleigh fading characterized by $g_{m, n}^{\rm NLoS}$. By exploiting the fact that a central chi-square random variable with two degrees of freedom is equivalent to an exponential distribution, we arrive at
\begin{flalign}
    & z_{1, 1} \sim {\rm Exp} \left( \frac{||{\bm \psi}_1^{\rm Pin} - {\bm \psi}_1^{\rm U}  ||^2 e^{2 \alpha ||{\bm \psi}_1^{\rm Pin} - {\bm \psi}_{1, 0}^{\rm Pin} || }}{  \eta p_1 \kappa^2 }  \right),
    \\
    & z_{2, 1} \sim {\rm Exp} \left( \frac{||{\bm \psi}_2^{\rm Pin} - {\bm \psi}_1^{\rm U}  ||^2 e^{2 \alpha ||{\bm \psi}_2^{\rm Pin} - {\bm \psi}_{2, 0}^{\rm Pin} || }}{\eta p_2 \kappa^2 \left( 2^{R_1} - 1 \right) }  \right).
\end{flalign}


By utilizing the cumulative distribution function (CDF) of the exponential distribution, we obtain
\begin{flalign}
    & \mathbb{P} \left( c_{2, 1}  \leq   z_{2, 1}  \right) = e^{ - \lambda_{2, 1} \max \left( c_{2, 1} , 0\right) }, \label{close_2} \\
    &  \mathbb{P} \left( z_{1, 1} \leq c_{1, 1}  \right) = 1 - e^{-\lambda_{1, 1} c_{1, 1} }, \label{close_3}
\end{flalign}
where 
\begin{flalign}
    & 
    \lambda_{1, 1}  \triangleq \frac{||{\bm \psi}_1^{\rm Pin} - {\bm \psi}_1^{\rm U}  ||^2 e^{2 \alpha ||{\bm \psi}_1^{\rm Pin} - {\bm \psi}_{1, 0}^{\rm Pin} || }}{  \eta p_1 \kappa^2 }, \\
    &
    c_{1, 1}  \triangleq \left( 2^{R_1} - 1 \right) \left( \frac{ \eta p_2  e^{-2 \alpha ||{\bm \psi}_2^{\rm Pin} - {\bm \psi}_{2, 0}^{\rm Pin} || } }{||{\bm \psi}_2^{\rm Pin} - {\bm \psi}_1^{\rm U}  ||^2}   + \sigma_1^2    \right), \\
    &
    \lambda_{2, 1} \triangleq 
    \frac{||{\bm \psi}_2^{\rm Pin} - {\bm \psi}_1^{\rm U}  ||^2 e^{2 \alpha ||{\bm \psi}_2^{\rm Pin} - {\bm \psi}_{2, 0}^{\rm Pin} || }}{\eta p_2 \kappa^2 \left( 2^{R_1} - 1 \right) },
    \\ &
    c_{2, 1} \triangleq \frac{ \eta p_1  e^{-2 \alpha ||{\bm \psi}_1^{\rm Pin} - {\bm \psi}_{1, 0}^{\rm Pin} || } }{||{\bm \psi}_1^{\rm Pin} - {\bm \psi}_1^{\rm U}  ||^2} - \left( 2^{R_1} - 1 \right) \sigma_1^2.
\end{flalign}

To derive a closed-form expression for the outage probability $\mathbb{P}( z_{1, 1} \leq z_{2, 1} + ( 2^{R_1} - 1 )\sigma_1^2 )$, the following lemma is introduced.

\begin{Lemma}
\label{exp_dis_2}
    Let $x_1 \sim {\rm Exp}\left( \lambda_1 \right)$ and $x_2 \sim {\rm Exp} \left( \lambda_2 \right)$ be independent exponential random variables with rate parameters $\lambda_1$ and $\lambda_2$, respectively. For any constant $c \geq 0$, the following probability identity holds:
    \begin{flalign}
        \mathbb{P} \left(  x_1 \leq x_2 +c \right) = 
         1   -   \frac{\lambda_2}{ \lambda_2 + \lambda_1}e^{ -\lambda_1 c }.
    \end{flalign}
\end{Lemma}
\begin{IEEEproof}
   Using the joint probability density function (PDF) of independent exponential random variables, we have
    \begin{flalign}
         & \mathbb{P} \left(  x_1 \leq x_2 +c \right) =  \int_{x_2 = 0}^{\infty} \int_{x_1 = 0}^{x_2 + c} \lambda_1 \lambda_2 e^{ - \lambda_1 x_1} e^{ - \lambda_2 x_2}  {\rm d} x_1 {\rm d} x_2 \nonumber \\
         & =  
        \int_{x_2 = 0}^{\infty} \lambda_2  e^{ - \lambda_2 x_2} \left(  \int_{x_1 = 0}^{x_2 + c} \lambda_1  e^{ - \lambda_1 x_1}  {\rm d} x_1 \right) {\rm d} x_2
         \nonumber \\
         &
         = \int_{x_2 = 0}^{\infty} \lambda_2  e^{ - \lambda_2 x_2} \left( 1 - e^{-\lambda_1 \left( x_2 + c \right) }
        \right) {\rm d} x_2
        \nonumber \\
        &
         = 1 - \frac{\lambda_2}{ \lambda_2 + \lambda_1 } e^{- \lambda_1 c}, 
    \end{flalign}
    which concludes the proof.
\end{IEEEproof}

By applying Lemma \ref{exp_dis_2}, we can obtain that
\begin{flalign}
    & \mathbb{P} \left( z_{1, 1} \leq z_{2, 1} + \left( 2^{R_1} - 1 \right)\sigma_1^2  \right) \nonumber \\
    &~~~~= 1 - \frac{\lambda_{2, 1}}{ \lambda_{2, 1} + \lambda_{1, 1} } e^{ - \lambda_{1, 1} \left( 2^{R_1} - 1 \right)\sigma_1^2  }. \label{close_4}
\end{flalign}

Combining \eqref{close_2}, \eqref{close_3}, and \eqref{close_4}, we derive a closed-form expression for $\mathbb{P}(\log_2 (1 + \Gamma_1) \leq R_1 )$, which is denoted by $\Phi_1 (p_1, p_2, x_1^{\rm Pin}, x_2^{\rm Pin}, R_1)$ and presented in \eqref{closed_u1}. Following a similar derivation procedure, the closed-form expression for $\mathbb{P}( \log_2(1+\Gamma_2) \leq R_2)$ can be established, which is denoted by $\Phi_2(p_1, p_2, x_1^{\rm Pin}, x_2^{\rm Pin}, R_2)$ and given in \eqref{closed_u2},
\begin{figure*}
\begin{flalign}
    & \Phi_1 \left( p_1, p_2, x_1^{\rm Pin}, x_2^{\rm Pin}, R_1 \right) = e^{-\beta || {\bm \psi}_1^{\rm Pin} - {\bm \psi}_1^{\rm U}||^2 }  e^{ - \beta || {\bm \psi}_2^{\rm Pin} - {\bm \psi}_1^{\rm U}  ||^2 } F_1 \left( p_1, p_2, x_1^{\rm Pin}, x_2^{\rm Pin}, R_1  \right) \nonumber \\
    & + e^{ - \beta || {\bm \psi}_1^{\rm Pin} - {\bm \psi}_1^{\rm U}||^2 } \left( 1 - e^{ -\beta || {\bm \psi}_2^{\rm Pin} - {\bm \psi}_1^{\rm U}  ||^2 }  \right)  e^{ - \lambda_{2, 1} \max \left( c_{2, 1}, 0\right) }   + \left( 1 - e^{ - \beta || {\bm \psi}_1^{\rm Pin} - {\bm \psi}_1^{\rm U}  ||^2  }  \right) e^{ - \beta || {\bm \psi}_2^{\rm Pin} - {\bm \psi}_1^{\rm U}  ||^2 } \left( 1 - e^{-\lambda_{1, 1} c_{1, 1} } \right) \nonumber \\
    & + \left( 1 - e^{-\beta || {\bm \psi}_1^{\rm Pin} - {\bm \psi}_1^{\rm U}  ||^2 }  \right) \left(  1 -  e^{-\beta || {\bm \psi}_2^{\rm Pin} - {\bm \psi}_1^{\rm U}   ||^2} \right)\left( 1 - \frac{\lambda_{2, 1}}{ \lambda_{2, 1} + \lambda_{1, 1} } e^{ - \lambda_{1, 1} \left( 2^{R_1} - 1 \right)\sigma_1^2  }  \right) \label{closed_u1} \\
     &
     \Phi_2\left(p_1, p_2, x_1^{\rm Pin}, x_2^{\rm Pin}, R_2\right) = 
   e^{- \beta || {\bm \psi}_1^{\rm Pin} - {\bm \psi}_2^{\rm U}  ||^2 } e^{-\beta || {\bm \psi}_2^{\rm Pin} - {\bm \psi}_2^{\rm U} ||^2 } F_2 \left( p_1, p_2, x_1^{\rm Pin}, x_2^{\rm Pin}, R_2  \right) \nonumber \\
   & + 
   e^{- \beta || {\bm \psi}_1^{\rm Pin} - {\bm \psi}_2^{\rm U}||^2 }
   \left( 1 - e^{-\beta || {\bm \psi}_2^{\rm Pin} - {\bm \psi}_2^{\rm U}||^2 } \right) \left( 1 - e^{ -\lambda_{2, 2} c_{2, 2} }  \right)   +
    \left( 1 - e^{ - \beta || {\bm \psi}_1^{\rm Pin} - {\bm \psi}_2^{\rm U}  ||^2 }  \right)
    e^{ - \beta || {\bm \psi}_2^{\rm Pin} - {\bm \psi}_2^{\rm U} ||^2 } e^{-\lambda_{1, 2} \max \left(c_{1, 2} , 0 \right)} \nonumber \\
    &  + 
    \left( 1 - e^{-\beta || {\bm \psi}_1^{\rm Pin} - {\bm \psi}_2^{\rm U}||^2}  \right)  \left( 1 - e^{ - \beta || {\bm \psi}_2^{\rm Pin} - {\bm \psi}_2^{\rm U} ||^2 }  \right)
    \left( 1 - \frac{\lambda_{1, 2}}{\lambda_{1, 2} + \lambda_{2, 2}} e^{- \lambda_{2, 2} \left(2^{R_2} - 1\right) \sigma_2^2 }
    \right)  \label{closed_u2}
\end{flalign}    
\hrule
\end{figure*}
where
\begin{flalign}
    & F_2 \left( p_1, p_2, x_1^{\rm Pin}, x_2^{\rm Pin}, R_2 \right)  \triangleq \nonumber \\
    & \begin{cases}
        1,  & \text{if}~ \frac{ \eta p_2 e^{-2 \alpha ||{\bm \psi}_2^{\rm Pin} - {\bm \psi}_{2, 0}^{\rm Pin} || } }{||{\bm \psi}_2^{\rm Pin} - {\bm \psi}_2^{\rm U}  ||^2} \\
         &~~
         \leq \left( 2^{R_2} - 1 \right) \left(  \frac{ \eta p_1 e^{-2 \alpha ||{\bm \psi}_1^{\rm Pin} - {\bm \psi}_{1, 0}^{\rm Pin} || } }{||{\bm \psi}_1^{\rm Pin} - {\bm \psi}_2^{\rm U}  ||^2}   + \sigma_2^2 \right), \\
        0,  & \text{Otherwise} , 
    \end{cases}  \\
    &
    \lambda_{1, 2}  \triangleq
    \frac{e^{2 \alpha ||{\bm \psi}_1^{\rm Pin} - {\bm \psi}_{1, 0}^{\rm Pin} || } ||{\bm \psi}_1^{\rm Pin} - {\bm \psi}_2^{\rm U}  ||^2}{ \eta p_1 \kappa^2 \left( 2^{R_2} - 1 \right)  },  \\
    &
    c_{1, 2} \triangleq
    \frac{ \eta p_2  e^{-2 \alpha ||{\bm \psi}_2^{\rm Pin} - {\bm \psi}_{2, 0}^{\rm Pin} || } }{||{\bm \psi}_2^{\rm Pin} - {\bm \psi}_2^{\rm U}  ||^2} - \left( 2^{R_2} - 1 \right) \sigma_2^2, \\
    &
    \lambda_{2, 2} \triangleq
    \frac{e^{2 \alpha ||{\bm \psi}_2^{\rm Pin} - {\bm \psi}_{2, 0}^{\rm Pin} || } ||{\bm \psi}_2^{\rm Pin} - {\bm \psi}_2^{\rm U}  ||^2 }{ \eta p_2 \kappa^2}, \\
    & 
    c_{2, 2} \triangleq
    \left( 2^{R_2} - 1 \right) \left(  \frac{ \eta p_1 e^{-2 \alpha ||{\bm \psi}_1^{\rm Pin} - {\bm \psi}_{1, 0}^{\rm Pin} || } }{||{\bm \psi}_1^{\rm Pin} - {\bm \psi}_2^{\rm U}  ||^2} + \sigma_2^2   \right).
\end{flalign}

Based on the preceding outage probability analysis, the considered optimization problem for the two-user two-PA scenario can be equivalently reformulated as 
\begin{subequations}
	\begin{flalign}
		{{\bf P}_1}: & \mathop{\max}  \limits_{\{ x_1^{\rm Pin}, x_2^{\rm Pin}, p_1, p_2, R_1, R_2 \} } R_1 + R_2  
			     \\
		{\rm s.t.} ~ & \Phi_1 \left( p_1, p_2, x_1^{\rm Pin}, x_2^{\rm Pin}, R_1 \right) \leq \epsilon, \label{outage_u1}  \  
       \\
       &  \Phi_2\left(p_1, p_2, x_1^{\rm Pin}, x_2^{\rm Pin}, R_2\right)  \leq \epsilon, \label{outage_u2}
       \\
        & p_1 + p_2 \leq P_{\max}, p_1 \geq 0, p_2 \geq 0,\label{P1:d}
         \\
        &  0 \leq x_1^{\rm Pin} \leq D_{\rm L},  0 \leq x_2^{\rm Pin} \leq D_{\rm L}.\label{P1:e}
	\end{flalign}
\end{subequations}

\subsection{PGD-based Algorithm}

It is observed that the functions $\Phi_1 ( p_1, p_2, x_1^{\rm Pin}, x_2^{\rm Pin}, R_1 )$ and $\Phi_2 ( p_1, p_2, x_1^{\rm Pin}, x_2^{\rm Pin}, R_2 )$ are monotonically increasing with respect to $R_1$ and $R_2$, respectively. Accordingly, the outage constraints \eqref{outage_u1} and \eqref{outage_u2} hold with equality at the optimal solution, which implies that
\begin{flalign}
    \Phi_1 & \left( p_1, p_2, x_1^{\rm Pin}, x_2^{\rm Pin}, R_1 \right) = \epsilon  \nonumber  \\
    & \implies R_1 = \Psi_1 \left( p_1, p_2, x_1^{\rm Pin}, x_2^{\rm Pin} \right), \\
    \Phi_2 & \left(p_1, p_2, x_1^{\rm Pin}, x_2^{\rm Pin}, R_2\right) =  \epsilon \nonumber \\
    & \implies R_2 = \Psi_2 \left( p_1, p_2, x_1^{\rm Pin}, x_2^{\rm Pin}  \right),
\end{flalign}
where $\Psi_1( \cdot )$ and $\Psi_2(\cdot)$ represent the corresponding implicit functions derived from \eqref{outage_u1} and \eqref{outage_u2}, respectively.

Then, Problem ${\bf P}_1$ can be equivalently reformulated as
\begin{subequations}
	\begin{flalign}
		{{\bf P}_2}: & \mathop{\max}  \limits_{\{ x_1^{\rm Pin}, x_2^{\rm Pin}, p_1, p_2 \} } \Psi_1 \left( p_1, p_2, x_1^{\rm Pin}, x_2^{\rm Pin} \right) \nonumber \\
        &~~~~~~~~~~~~~~~
        + \Psi_2 \left( p_1, p_2, x_1^{\rm Pin}, x_2^{\rm Pin} \right) 
			     \\
		{\rm s.t.} ~&\eqref{P1:d},\eqref{P1:e}.
	\end{flalign}
\end{subequations}
Although Problem ${\bf P}_2$ involves a highly non-convex and complicated objective function, its constraint set exhibits a tractable geometric structure. This asymmetric difficulty between the objective function and the constraints naturally motivates the adoption of a PGD-based algorithm. By leveraging PGD, we can efficiently handle the non-convexity through gradient updates, while exploiting the simplicity of the feasible region to perform computationally inexpensive projections.

To implement the PGD-based algorithm, we first need to compute the gradients of $\Psi_1(\cdot)$ and $\Psi_2(\cdot)$ with respect to the optimization variables. Let $x \in \{p_1, p_2, x_1^{\rm Pin}, x_2^{\rm Pin}\}$, by applying the implicit function theorem \cite{implicit}, we can derive the following expressions\footnote{Note that due to the non-smooth nature of the indicator functions $F_1(\cdot)$ and $F_2(\cdot)$ within $\Psi_1(\cdot)$ and $\Psi_2(\cdot)$, conventional gradients do not exist at these non-differentiable points. Therefore, the standard gradient calculations are rigorously replaced by their subgradient counterparts.}:
\begin{flalign}
     & \frac{\partial \Psi_1 \left( p_1, p_2, x_1^{\rm Pin}, x_2^{\rm Pin} \right)} {\partial x} = - \frac{\frac{\partial  \Phi_1  \left( p_1, p_2, x_1^{\rm Pin}, x_2^{\rm Pin}, R_1 \right)} {\partial x}}   { \frac{\partial  \Phi_1  \left( p_1, p_2, x_1^{\rm Pin}, x_2^{\rm Pin}, R_1 \right)} {\partial R_1} },  \\
     & \frac{\partial \Psi_2 \left(p_1, p_2, x_1^{\rm Pin}, x_2^{\rm Pin}  \right)}{\partial x} = - \frac{\frac{\partial \Phi_2 \left(p_1, p_2, x_1^{\rm Pin}, x_2^{\rm Pin}, R_2 \right) }{\partial x}}{ \frac{\partial \Phi_2 \left( p_1, p_2, x_1^{\rm Pin}, x_2^{\rm Pin}, R_2 \right) }{ \partial R_2}   }.
\end{flalign}

The proposed PGD-based algorithm is executed in an iterative manner. In the $i$-th iteration, let $\{ \widetilde{x}_{1, i}^{\rm Pin}, \widetilde{x}_{2, i}^{\rm Pin}, \widetilde{p}_{1, i}, \widetilde{p}_{2, i} \}$ denote the corresponding iteration point of Problem ${\bf P}_2$. Then, the gradient descent update is performed as
\begin{flalign}
    & \bar{x}_{1, i}^{\rm Pin} = \widetilde{x}_{1, i}^{\rm Pin}  + \gamma \left. \frac{\partial {\Psi}_1 \left( p_1, p_2, x_1^{\rm Pin}, x_2^{\rm Pin} \right) } {\partial x_1^{\rm Pin}} \right|_{\{ \widetilde{p}_{1, i}, \widetilde{p}_{2, i}, \widetilde{x}_{1, i}^{\rm Pin}, \widetilde{x}_{2, i}^{\rm Pin} \}} \nonumber \\
    & ~~~~  + \gamma \left. \frac{ \partial \Psi_2 \left( p_1, p_2, x_1^{\rm Pin}, x_2^{\rm Pin} \right) }{ \partial x_1^{\rm Pin} } \right|_{\{ \widetilde{p}_{1, i}, \widetilde{p}_{2, i}, \widetilde{x}_{1, i}^{\rm Pin}, \widetilde{x}_{2, i}^{\rm Pin} \}}, \label{x1_gd} \\
    & \bar{x}_{2, i}^{\rm Pin} = \widetilde{x}_{2, i}^{\rm Pin} + \left. \gamma \frac{\partial \Psi_1 \left( p_1, p_2, x_1^{\rm Pin}, x_2^{\rm Pin} \right)  }{ \partial x_2^{\rm Pin}} \right|_{\{ \widetilde{p}_{1, i}, \widetilde{p}_{2, i}, \widetilde{x}_{1, i}^{\rm Pin}, \widetilde{x}_{2, i}^{\rm Pin} \} }  \nonumber  \\
    & ~~~~ + \gamma \left.  \frac{\partial \Psi_2 \left( p_1, p_2, x_1^{\rm Pin}, x_2^{\rm Pin} \right) }{ \partial x_2^{\rm Pin}} \right|_{\{ \widetilde{p}_{1, i}, \widetilde{p}_{2, i}, \widetilde{x}_{1, i}^{\rm Pin}, \widetilde{x}_{2, i}^{\rm Pin} \}}, \\
    & \bar{p}_{1, i} = \widetilde{p}_{1, i} + \left. \gamma \frac{\partial \Psi_1 \left( p_1, p_2, x_1^{\rm Pin}, x_2^{\rm Pin} \right)}{\partial p_1} \right|_{\{\widetilde{p}_{1, i}, \widetilde{p}_{2, i}, \widetilde{x}_{1, i}^{\rm Pin}, \widetilde{x}_{2, i}^{\rm Pin}\}}   \nonumber \\
    & ~~~~ + \gamma \left. \frac{\partial \Psi_2 \left( p_1, p_2, x_1^{\rm Pin}, x_2^{\rm Pin} \right) }{\partial p_1} \right|_{\{ \widetilde{p}_{1, i}, \widetilde{p}_{2, i}, \widetilde{x}_{1, i}^{\rm Pin}, \widetilde{x}_{2, i}^{\rm Pin}  \} } ,  \\
    & \bar{p}_{2, i} = \widetilde{p}_{2, i} + \gamma \left. \frac{\partial \Psi_1 \left(p_1, p_2, x_1^{\rm Pin}, x_2^{\rm Pin} \right) }{\partial p_2} \right|_{\{ \widetilde{p}_{1, i}, \widetilde{p}_{2, i}, \widetilde{x}_{1, i}^{\rm Pin}, \widetilde{x}_{2, i}^{\rm Pin} \}}  \nonumber \\
    & ~~~~ + \gamma \left. \frac{\partial \Psi_2 \left( p_1, p_2, x_1^{\rm Pin}, x_2^{\rm Pin} \right)}{\partial p_2} \right|_{\{\widetilde{p}_{1, i}, \widetilde{p}_{2, i}, \widetilde{x}_{1, i}^{\rm Pin}, \widetilde{x}_{2, i}^{\rm Pin}  \}},  \label{p2_gd}
\end{flalign}
where $\gamma$ denotes the step size.

After completing the gradient descent step, we project the intermediate iterate $\{ \bar{x}_{1, i}^{\rm Pin}, \bar{x}_{2, i}^{\rm Pin}, \bar{p}_{1, i}, \bar{p}_{2, i} \}$ onto the feasible region of Problem ${\bf P}_2$ to ensure that all constraints are strictly satisfied. Specifically, we seek a feasible point $\{x_1^{\rm Pin}, x_2^{\rm Pin}, p_1, p_2  \}$ within the feasible set that lies closest to the intermediate iterate point. This projection task can be formulated as two independent optimization problems, which are respectively given by
\begin{subequations}
\begin{flalign}
    {\bf P}_3 : & \mathop{\min}  \limits_{ \{ x_1^{\rm Pin}, x_2^{\rm Pin} \} }  \left( x_1^{\rm Pin} - \bar{x}_{1, i}^{\rm Pin} \right)^2 + \left( x_2^{\rm Pin} - \bar{x}_{2, i}^{\rm Pin} \right)^2  \\
    {\rm s.t.}~& \eqref{P1:e}.
\end{flalign}
\end{subequations}
\begin{subequations}
\begin{flalign}
    {\bf P}_4 : & \mathop{\min}  \limits_{ \{ p_1, p_2 \} } \left(  p_1 - \bar{p}_{1, i} \right)^2 + \left( p_2 - \bar{p}_{2, i}  \right)^2    \\
    {\rm s.t.}~& \eqref{P1:d}.
\end{flalign}
\end{subequations}

We begin by solving Problem ${\bf P}_3$. By exploiting the unimodal property of the quadratic function, the optimal solution to Problem ${\bf P}_3$, denoted by $\{ \widetilde{x}_{1, i+1}^{\rm Pin}, \widetilde{x}_{2, i+1}^{\rm Pin} \}$, is given by
\begin{flalign}
    \widetilde{x}_{1, i+1}^{\rm Pin} = \begin{cases}
        0, & \bar{x}_{1, i}^{\rm Pin} < 0 \\
        \bar{x}_{1, i}^{\rm Pin}, & 0 \leq \bar{x}_{1, i}^{\rm Pin} \leq D_{\rm L}   \\
        D_{\rm L}, & \bar{x}_{1, i}^{\rm Pin} > D_{\rm L}
    \end{cases},  \label{x1_update}
\end{flalign}
\begin{flalign}
\widetilde{x}_{2, i+1}^{\rm Pin} = 
    \begin{cases}
        0, & \bar{x}_{2, i}^{\rm Pin} < 0 \\
        \bar{x}_{2, i}^{\rm Pin}, & 0 \leq \bar{x}_{2, i}^{\rm Pin} \leq D_{\rm L} \\
        D_{\rm L}, & \bar{x}_{2, i}^{\rm Pin} > D_{\rm L}
    \end{cases}. \label{x2_update}
\end{flalign}

As for Problem ${\bf P}_4$, it is noted that the considered problem is convex, and thus its optimal solution can be rigorously derived by invoking the KKT conditions. To this end, we construct the Lagrangian function for Problem ${\bf P}_4$, which is given by
\begin{flalign}
    \mathcal{L} = & (p_1 - \bar{p}_{1, i} )^2 + (p_2 - \bar{p}_{2, i})^2 + \lambda (p_1 + p_2 - P_{\max}) \nonumber \\
    & - \mu_1 p_1 - \mu_2 p_2,
\end{flalign}
where $\lambda$, $\mu_1$, and $\mu_2$ are the corresponding Lagrange multipliers. The KKT conditions for Problem ${\bf P}_4$ are given by
\begin{flalign}
\begin{cases}
    p_1 + p_2 \leq P_{\max},~ p_1 \geq 0,~ p_2 \geq 0, \\
    \lambda \geq 0, ~ \mu_1 \geq 0, ~ \mu_2 \geq 0,  \\ 
    \lambda (p_1 + p_2 - P_{\max}) = 0, ~ \mu_1 p_1 = 0, ~ \mu_2 p_2 = 0, \\
    2 ( p_1 - \bar{p}_{1, i} ) + \lambda - \mu_1 = 0, \\
    2 ( p_2 - \bar{p}_{2, i} ) + \lambda - \mu_2 = 0.
\end{cases} \label{kkt}
\end{flalign}
Since the Lagrange multipliers $\lambda$, $\mu_1$, and $\mu_2$ are constrained to be non-negative, the KKT conditions in \eqref{kkt} can be classified into $8$ independent cases. We then examine each case in detail to derive the closed-form optimal solution and identify the corresponding boundary conditions.

\subsubsection{Case with $\lambda = 0$, $\mu_1 = 0$, and $\mu_2=0$} Under this case, the optimal power allocation is 
\begin{flalign}
    p_1 = \bar{p}_{1, i}, ~ p_2 = \bar{p}_{2, i},
\end{flalign}
with the following boundary conditions:
\begin{flalign*}
    \bar{p}_{1, i} + \bar{p}_{2, i} \leq P_{\max},~ \bar{p}_{1, i} \geq 0, ~ \bar{p}_{2, i} \geq 0.
\end{flalign*}

\subsubsection{Case with $\lambda = 0$, $\mu_1 = 0$, and $\mu_2 > 0$} Under this case, the optimal power allocation is
\begin{flalign}
    p_1 = \bar{p}_{1, i}, ~ p_2 = 0,
\end{flalign}
with the following boundary conditions:
\begin{flalign*}
    0 \leq \bar{p}_{1, i} \leq P_{\max},~ \bar{p}_{2, i} < 0.
\end{flalign*}

\subsubsection{Case with $\lambda = 0$, $\mu_1 > 0$, and $\mu_2 = 0$} Under this case, the optimal power allocation is
\begin{flalign}
    p_1 = 0, ~ p_2 = \bar{p}_{2, i},
\end{flalign}
with the following boundary conditions:
\begin{flalign*}
    \bar{p}_{1, i} < 0,~ 0 \leq \bar{p}_{2, i} \leq P_{\max}.
\end{flalign*}

\subsubsection{Case with $\lambda > 0$, $\mu_1 = 0$, and $\mu_2 = 0$} Under this case, the optimal power allocation is
\begin{flalign}
    p_1 = \frac{P_{\max} + \bar{p}_{1, i} -\bar{p}_{2, i} }{2}, ~p_2  = \frac{P_{\max} - \bar{p}_{1, i} + \bar{p}_{2, i} }{2},
\end{flalign}
with the following boundary conditions:
\begin{flalign*}
    \bar{p}_{1, i} + \bar{p}_{2, i} > P_{\max}, ~ \bar{p}_{2, i} - \bar{p}_{1, i} < P_{\max},~ \bar{p}_{1, i} - \bar{p}_{2, i} < P_{\max}.
\end{flalign*}

\subsubsection{Case with $\lambda = 0$, $\mu_1 > 0$, and $\mu_2 > 0$} Under this case, the optimal power allocation is
\begin{flalign}
    p_1 = 0, ~ p_2 = 0,
\end{flalign}
with the following boundary conditions:
\begin{flalign*}
    \bar{p}_{1, i} < 0,~\bar{p}_{2, i} < 0.
\end{flalign*}

\subsubsection{Case with $\lambda > 0$, $\mu_1 = 0$, and $\mu_2 > 0$}  Under this case, the optimal power allocation is
\begin{flalign}
    p_1 = P_{\max}, ~p_2 = 0,
\end{flalign}
with the following boundary conditions:
\begin{flalign*}
    \bar{p}_{1, i} > P_{\max}, ~ \bar{p}_{1, i} - \bar{p}_{2, i} > P_{\max}.
\end{flalign*}

\subsubsection{Case with $\lambda > 0$, $\mu_1 > 0$, and $\mu_2 = 0$} Under this case, the optimal power allocation is
\begin{flalign}
    p_1 = 0, ~ p_2  = P_{\max},
\end{flalign}
with the following boundary conditions:
\begin{flalign*}
     \bar{p}_{2, i} > P_{\max},~
     \bar{p}_{2, i} - \bar{p}_{1, i} > P_{\max}.   
\end{flalign*}

\subsubsection{Case with $\lambda > 0$, $\mu_1 > 0$, and $\mu_2 >0$} Under this case, the optimality conditions derived from \eqref{kkt} require
\begin{flalign}
    p_1 = 0,~p_2 = 0,~p_1 + p_2 = P_{\max},
\end{flalign}
which is infeasible for any $P_{\max} > 0$, implying that this case is physically invalid.

By synthesizing all $8$ cases, the optimal solution to Problem ${\bf P}_4$, denoted by $\{ \widetilde{p}_{1, i+1}, \widetilde{p}_{2, i+1} \}$, is given by \eqref{p1p2_update}.
\begin{figure*}
\begin{flalign}
    \label{p1p2_update}
    \{ \widetilde{p}_{1, i+1}, \widetilde{p}_{2, i+1} \} = \begin{cases}
        \{ \bar{p}_{1, i}, \bar{p}_{2, i} \}, &  \bar{p}_{1, i} + \bar{p}_{2, i} \leq P_{\max},~ \bar{p}_{1, i} \geq 0, ~ \bar{p}_{2, i} \geq 0. \\
        \{ \bar{p}_{1, i}, 0 \}, & 0 \leq \bar{p}_{1, i} \leq P_{\max},~ \bar{p}_{2, i} < 0.  \\
        \{ 0, \bar{p}_{2, i} \}, & \bar{p}_{1, i} < 0,~ 0 \leq \bar{p}_{2, i} \leq P_{\max}. \\
        \left\{ \frac{P_{\max} + \bar{p}_{1, i} -\bar{p}_{2, i} }{2}, \frac{P_{\max} - \bar{p}_{1, i} + \bar{p}_{2, i} }{2}  \right\}, & \bar{p}_{1, i} + \bar{p}_{2, i} > P_{\max}, ~ \bar{p}_{2, i} - \bar{p}_{1, i} < P_{\max},~ \bar{p}_{1, i} - \bar{p}_{2, i} < P_{\max}. \\
        \{0, 0\}, & \bar{p}_{1, i} < 0,~\bar{p}_{2, i} < 0. \\
        \{P_{\max}, 0\}, & \bar{p}_{1, i} > P_{\max}, ~ \bar{p}_{1, i} - \bar{p}_{2, i} > P_{\max}  \\
        \{0, P_{\max}\}, & \bar{p}_{2, i} > P_{\max},~
     \bar{p}_{2, i} - \bar{p}_{1, i} > P_{\max}. 
    \end{cases}
\end{flalign}
\hrule
\end{figure*}

The overall procedure of the proposed PGD-based algorithm is summarized in Algorithm \ref{alg:1}.

  \begin{algorithm}
    \caption{The proposed PGD-based algorithm.} \label{alg:1}
\begin{algorithmic}
  \State{\textbf{Initialization}:}
  \State{~ Set $i:=0$ and obtain a feasible $\{ \widetilde{x}_{1, i}^{\rm Pin}, \widetilde{x}_{2, i}^{\rm Pin}, \widetilde{p}_{1, i}, \widetilde{p}_{2, i} \}$ to Problem ${\bf P}_2$;}
  \While{\rm the stop criterion is not satisfied}
      \State{Obtain $\{ \bar{x}_{1, i}^{\rm Pin}, \bar{x}_{2, i}^{\rm Pin}, \bar{p}_{1, i}, \bar{p}_{2, i} \}$ via \eqref{x1_gd}$-$\eqref{p2_gd};}
      \State{Obtain $\{ \widetilde{x}_{1, i+1}^{\rm Pin}, \widetilde{x}_{2, i+1}^{\rm Pin} \}$ via \eqref{x1_update} and \eqref{x2_update};}
      \State{Obtain $\{ \widetilde{p}_{1, i+1}, \widetilde{p}_{2, i+1} \}$ via \eqref{p1p2_update};}
      \State{Update $i := i+1$;}
  \EndWhile
\end{algorithmic}
\end{algorithm}

\textit{Remark 2:} The exact analytical derivations and the specialized PGD-based algorithm developed in this section for the foundational two-user two-PA scenario are intrinsically linked to the generalized multi-user multi-PA scenario presented in the next section. Specifically, the exact closed-form expressions for the outage probabilities derived herein serve as an indispensable theoretical benchmark to rigorously evaluate the tightness and accuracy of the approximations employed for the general case. Furthermore, the customized PGD-based algorithm leveraging KKT conditions can significantly reduce the execution time for the two-user two-PA scenario. These interconnected analytical and algorithmic insights will be comprehensively validated via numerical simulations in Section V.

\section{Joint Optimization of PA Deployment and Power Allocation}

In this section, we extend our analysis to a more general scenario with $N$ PAs servicing $N$ users. To facilitate practical design, we derive a closed-form approximation for the outage probabilities. On this basis, an SCA-based algorithm is further developed to handle the approximated optimization problem.

\subsection{Outage Probability Approximation}
For the general $N$-user $N$-PA scenario, enumerating all LoS/NLoS link combinations as in the preceding two-user two-PA case is computationally prohibitive, since the number of channel conditions scales as $2^{N^2}$. As verified in \cite{los_block}, interfering links $h_{m, n}~( m \neq n)$ are predominantly NLoS since the distance between the user and the PA in various sub-regions exceeds $D_{\rm W} / 2N$, making blockage highly probable. As illustrated in Fig. \ref{los_prob_dis}, the LoS existence probability falls below $1\%$ when the transmission distance exceeds $25~{\rm m}$. 
\begin{figure}
	\begin{center}
	\centerline{\includegraphics[ width=.46\textwidth]{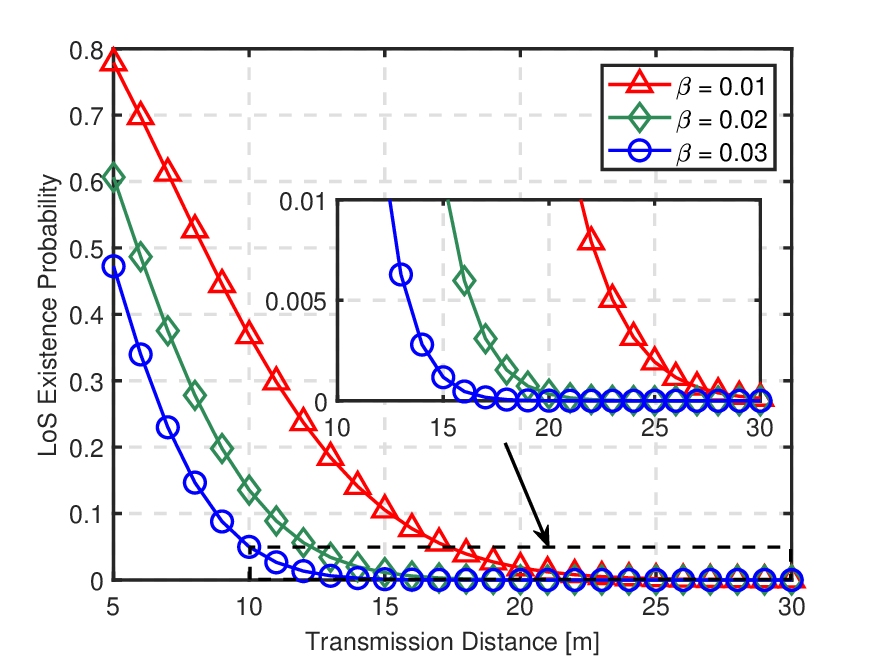}}
    \caption{LoS existence probability versus transmission distance with blockage coefficient $\beta \in \{0.01, 0.02, 0.03\}$.}
	\label{los_prob_dis}
    \vspace{-2em}
	\end{center}
\end{figure}
Therefore, we assume that all interfering links are NLoS, i.e., $\zeta_{m, n} = 0, ~\forall m \neq n$. Under this assumption, the outage probability in \eqref{p0b} can be approximated as
\begin{flalign}
     & \mathbb{P} \left( \log_2 \left( 1 + \Gamma_n \right) \leq R_n \right) \label{prob_appro} \\
     & \approx  \mathbb{P} \left( \zeta_{n, n} = 1  \right) \mathbb{P} \left( \Gamma_n \leq 2^{R_n} - 1 | \zeta_{n, n} = 1 \right) \nonumber \\
     & ~~ + \mathbb{P} \left( \zeta_{n, n} = 0 \right) \mathbb{P} \left( \Gamma_n \leq 2^{R_n} - 1 | \zeta_{n, n} = 0 \right) \nonumber \\
     & = e^{ - \beta ||{\bm \psi}_n^{\rm Pin} - {\bm \psi}_n^{\rm U} ||^2 }
      \mathbb{P} \left( a_n  \leq     \sum\nolimits_{m \neq n} z_{m, n}
      \right)  \nonumber \\
     & ~~ + \left( 1 - e^{ - \beta || {\bm \psi}_n^{\rm Pin} - {\bm \psi}_n^{\rm U}  ||^2 }   \right)
     \mathbb{P} \left( z_{n} \leq  \sum\nolimits_{m \neq n} z_{m, n} + b_n
    \right) ,  \nonumber
\end{flalign}
where 
\begin{subequations}
\begin{flalign}
    & a_n \triangleq \frac{\eta  p_n e^{ - 2 \alpha \left| \left| {\bm \psi}_n^{\rm Pin} - {\bm \psi}_{n, 0}^{\rm Pin}  \right| \right| }}{\left| \left| {\bm \psi}_{n}^{\rm Pin} - {\bm \psi}_n^{\rm U} \right| \right|^2 } - \left( 2^{R_n} - 1 \right) \sigma_n^2 , \\
    & b_n \triangleq \left( 2^{R_n} - 1 \right)\sigma_n^2 , \\
    & z_{m, n} \triangleq  \frac{\eta p_m \kappa^2 \left( 2^{R_n} - 1 \right) }{ e^{ 2 \alpha \left| \left| {\bm \psi}_m^{\rm Pin} - {\bm \psi}_{m, 0}^{\rm Pin}  \right| \right| } \left| \left| {\bm \psi}_{m}^{\rm Pin} - {\bm \psi}_n^{\rm U} \right| \right|^2 }  \left| g_{m, n}^{\rm NLoS}  \right|^2  \nonumber \\
    & ~~~~~~ \sim {\rm Exp}\left(  \frac{e^{ 2 \alpha \left| \left| {\bm \psi}_m^{\rm Pin} - {\bm \psi}_{m, 0}^{\rm Pin}  \right| \right| } \left| \left| {\bm \psi}_{m}^{\rm Pin} - {\bm \psi}_n^{\rm U} \right| \right|^2}{\eta p_m \kappa^2 \left( 2^{R_n} - 1 \right)} \right),
    \\
    & z_n \triangleq \frac{\eta p_n \kappa^2 e^{ - 2 \alpha \left| \left| {\bm \psi}_n^{\rm Pin} - {\bm \psi}_{n, 0}^{\rm Pin}  \right| \right| }}{\left| \left| {\bm \psi}_{n}^{\rm Pin} - {\bm \psi}_n^{\rm U} \right| \right|^2 }  \left| g_{n, n}^{\rm NLoS}  \right|^2 \nonumber \\
    & ~~~~~~ \sim {\rm Exp} \left( \frac{\left| \left| {\bm \psi}_{n}^{\rm Pin} - {\bm \psi}_n^{\rm U} \right| \right|^2}{\eta p_n \kappa^2 e^{ - 2 \alpha \left| \left| {\bm \psi}_n^{\rm Pin} - {\bm \psi}_{n, 0}^{\rm Pin}  \right| \right| }}  \right).
\end{flalign}
\end{subequations}

We first tackle the term $\mathbb{P}( a_n \leq \sum_{m \neq n}^N z_{m, n} )$. Note that the aggregate interference $\sum_{m \neq n}^N z_{m, n}$ obeys a hypo-exponential distribution, whose survival function is characterized in the following lemma.

\begin{Lemma}\label{hypo_exp_n}
    Let $\{ z_n\}_{n=1}^N$ denote a set of independent random variables with $z_n \sim {\rm Exp} (\lambda_n)$, where $\lambda_n$ denotes the corresponding rate parameter. For any given constant $c \geq 0$, the survival probability of $\sum_{n=1}^N z_n$ is given by
    \begin{flalign}
        \mathbb{P} \left( c \leq \sum_{n=1}^N z_n  \right) = \sum_{n=1}^N \left( \prod^N_{ k \neq n} \frac{\lambda_k }{ \lambda_k - \lambda_n}   \right) e^{ - \lambda_n c}.
    \end{flalign}
\end{Lemma}
\begin{IEEEproof}
    This result follows from the standard characterization of the hypo-exponential distribution. Please refer to Chapter 5 of \cite{hypo-exp} for details.
\end{IEEEproof}

Lemma \ref{hypo_exp_n} imposes the condition $a_n \geq 0$, which translates to the following constraint
\begin{flalign}
    R_n \leq \log_2 \left(  1 + \frac{\eta  p_n e^{ - 2 \alpha \left| \left| {\bm \psi}_n^{\rm Pin} - {\bm \psi}_{n, 0}^{\rm Pin}  \right| \right| }}{\left| \left| {\bm \psi}_{n}^{\rm Pin} - {\bm \psi}_n^{\rm U} \right| \right|^2  \sigma_n^2 } \right). \label{R_upper_bound}
\end{flalign}
Constraint \eqref{R_upper_bound} must hold for the considered problem. Otherwise, the target information rate $R_n$ exceeds the maximum achievable rate even when the desired link $h_{n, n}$ is LoS in the absence of interference. In this scenario, a rate outage event is unavoidable for all channel realizations. 

Combining Lemma \ref{hypo_exp_n} and constraint \eqref{R_upper_bound}, we obtain a closed-form expression for $\mathbb{P}(a_n \leq \sum^N_{m \neq n} z_{m, n} )$ as 
\begin{flalign}
     \mathbb{P} &\left( a_n \leq \sum\nolimits_{m \neq n}^N z_{m, n}  \right) \nonumber \\
     &= \sum\nolimits^N_{m \neq n} \Bigg( e^{  \frac{ \sigma_n^2 \Omega_{m, n} \left( x_m^{\rm Pin} \right) }{\eta p_m \kappa^2 }    - \frac{ p_n \Omega_{m, n} \left( x_m^{\rm Pin} \right) }{p_m \kappa^2 \left( 2^{R_n} - 1 \right) \Omega_{n, n} \left( x_n^{\rm Pin} \right)} } \nonumber \\
     & \times \left( \prod\nolimits^N_{k \neq m, n} \frac{ p_m \Omega_{k, n} \left( x_k^{\rm Pin} \right)   }{ p_m \Omega_{k, n} \left( x_k^{\rm Pin} \right) - p_k \Omega_{m, n} \left( x_m^{\rm Pin} \right)  }  \right)  \Bigg),  \label{los_prob_appro}
\end{flalign}
where
\begin{flalign}
    \Omega_{m, n} \left( x_m^{\rm Pin} \right) \triangleq  \left| \left| {\bm \psi}_m^{\rm Pin} - {\bm \psi}_n^{\rm U}  \right| \right|^2 e^{2 \alpha \left|\left| {\bm \psi}_m^{\rm Pin} - {\bm \psi}_{m, 0}^{\rm Pin}   \right|\right| }. 
\end{flalign}

We next evaluate the term $\mathbb{P}(z_n \leq \sum_{m \neq n}^N z_{m, n} + b_n )$ by resorting to the following lemma. 

\begin{Lemma} \label{prob_nlos}
Let $\{ z_n\}_{n=1}^N$ be a set of independent random variables with $z_n \sim {\rm Exp} (\lambda_n)$, where $\lambda_n$ denotes the corresponding rate parameter. For any given constant $c \geq 0$, the following equation holds:
\begin{flalign}
    \mathbb{P}\left( z_n \leq \sum^N_{m \neq n} z_m + c  \right) = 1 - e^{ - \lambda_n c} \left(\prod^N_{m \neq n} \frac{\lambda_m}{\lambda_m + \lambda_n}   \right). \label{lemma3_result}
\end{flalign}
\end{Lemma}
\begin{IEEEproof}
    Let $y_n = \sum^N_{m \neq n} z_m$. By definition, $y_n$ follows a hypo-exponential distribution, whose probability density function is given by
    \begin{flalign*}
        f_{y_n} (y) = \begin{cases}
            \sum_{m \neq n}^N \lambda_m e^{ - \lambda_m y} \left( \prod^N_{k \neq m, n} \frac{\lambda_k} {\lambda_k - \lambda_m } \right), & y\geq 0, \\
            0, & y < 0 .
        \end{cases}
    \end{flalign*}
    Then, the probability in \eqref{lemma3_result} can be evaluated via the joint probability density function of $(z_n, y_n)$, which yields
    \begin{flalign}
        & \mathbb{P} \left( z_n \leq \sum^N_{m \neq n} z_{m} + c  \right) = \mathbb{P} \left( z_n \leq y_n + c  \right) = \int_{y = 0}^{\infty}  \Bigg[ \int_{x = 0}^{y + c}  \nonumber \\
        &  \lambda_n e^{-\lambda_n x} {\rm d} x \Bigg] \sum^N_{m \neq n}  \lambda_m  e^{ - \lambda_m y} \left( \prod^N_{k \neq m, n} \frac{\lambda_k} {\lambda_k - \lambda_m } \right)  {\rm d} y \nonumber \\
        & = \sum^N_{m \neq n} \left( \prod^N_{ k \neq m, n} \frac{\lambda_k}{\lambda_k - \lambda_m} \right) \left(  1 -  e^{-\lambda_n c} \frac{\lambda_m}{\lambda_m + \lambda_n}  \right)  \nonumber \\
        & \overset{(a)}{=}  1 -   e^{-\lambda_n c} \sum^N_{m \neq n} \left( \prod^N_{ k \neq m, n} \frac{\lambda_k}{\lambda_k - \lambda_m} \right) \frac{\lambda_m}{\lambda_m + \lambda_n} \nonumber \\
        & \overset{(b)}{=} 1 - e^{-\lambda_n c} \left( \prod^N_{m \neq n} \frac{\lambda_m}{\lambda_m + \lambda_n}  \right),
    \end{flalign} 
    where  $(a)$ is due to the normalization property of the hypo-exponential distribution, and $(b)$ is due to the Laplace transform of the hypo-exponential distribution. 
\end{IEEEproof}

By applying Lemma \ref{prob_nlos}, we derive a closed-form expression for $\mathbb{P}( z_n \leq \sum_{m \neq n}^N z_{m, n} + b_n )$, which is given by
\begin{flalign}
    & \mathbb{P} \left( z_n \leq \sum\nolimits^N_{m \neq n} z_{m, n} + b_n  \right) = 1 - e^{- \frac{ \left( 2^{R_n} - 1 \right)\sigma_n^2 \Omega_{n, n} \left( x_n^{\rm Pin} \right) }{\eta p_n \kappa^2 } }  \nonumber  \\
    & \times
    \left( \prod^N\nolimits_{m \neq n} \frac{ p_n \Omega_{m, n} \left( x_m^{\rm Pin} \right)  }{ p_n \Omega_{m, n} \left( x_m^{\rm Pin} \right)  + p_m \left( 2^{R_n} - 1 \right) \Omega_{n, n} \left( x_n^{\rm Pin} \right) }  \right). \label{nlos_prob_appro}
\end{flalign}

Substituting \eqref{los_prob_appro} and \eqref{nlos_prob_appro} into \eqref{prob_appro} yields a tractable approximation for the outage probability $\mathbb{P}(\log_2(1 + \Gamma_n) \leq R_n )$, which is denoted by $\Phi_n (\{p_m, x_m^{\rm Pin}, R_m \})$ and given in \eqref{phi_n}.
\begin{figure*}
\begin{flalign}
    & \Phi_n \left( \{ p_m, x_m^{\rm Pin}, R_m  \}  \right) \nonumber \\  
    & \triangleq
    e^{ - \beta ||{\bm \psi}_n^{\rm Pin} - {\bm \psi}_n^{\rm U} ||^2 }
      \sum\nolimits_{m \neq n} \Bigg( e^{  \frac{ \sigma_n^2 \Omega_{m, n} \left( x_m^{\rm Pin} \right) }{\eta p_m \kappa^2 }    - \frac{ p_n \Omega_{m, n} \left( x_m^{\rm Pin} \right) }{p_m \kappa^2 \left( 2^{R_n} - 1 \right) \Omega_{n, n} \left( x_n^{\rm Pin} \right)} } \left( \prod\nolimits_{k \neq m, n} \frac{ p_m \Omega_{k, n} \left( x_k^{\rm Pin} \right)   }{ p_m \Omega_{k, n} \left( x_k^{\rm Pin} \right) - p_k \Omega_{m, n} \left( x_m^{\rm Pin} \right)  }  \right)  \Bigg)
     \nonumber \\
     & ~~ + \left( 1 - e^{ - \beta || {\bm \psi}_n^{\rm Pin} - {\bm \psi}_n^{\rm U}  ||^2 }   \right)
     \left(  1 - e^{- \frac{ \left( 2^{R_n} - 1 \right)\sigma_n^2 \Omega_{n, n} \left( x_n^{\rm Pin} \right) }{\eta p_n \kappa^2 } } 
    \left( \prod\nolimits_{m \neq n} \frac{ p_n \Omega_{m, n} \left( x_m^{\rm Pin} \right)  }{ p_n \Omega_{m, n} \left( x_m^{\rm Pin} \right)  + p_m \left( 2^{R_n} - 1 \right) \Omega_{n, n} \left( x_n^{\rm Pin} \right) }  \right)  \right)  \label{phi_n}
\end{flalign}
\hrule
\end{figure*}

Consequently, Problem ${\bf P}_0$ can be approximated by the following optimization problem:
\begin{subequations}
	\begin{flalign}
		{{\bf P}_5}: & \mathop{\max}  \limits_{\{ x_n^{\rm Pin}, p_n, R_n \} } \sum\nolimits_{n=1}^N R_n   
			  \label{p5a}   \\
		{\rm s.t.} ~ & \Phi_n \left( \{ p_m, x_m^{\rm Pin}, R_m \} \right) \leq \epsilon,  \forall n \in \{1, 2, \ldots, N\}, \label{p5b}  \\
        & \eqref{p0c},~\eqref{p0d},~\eqref{R_upper_bound}. \nonumber
	\end{flalign}
\end{subequations}
Problem ${\bf P}_5$ is non-convex due to constraints \eqref{R_upper_bound} and \eqref{p5b}. Furthermore, the complicated coupling among $\{ x_n^{\rm Pin} \}$, $\{p_n\}$, and $\{R_n\}$ makes the considered problem more intractable.

\subsection{SCA-based Algorithm}
We address the non-convex constraint \eqref{R_upper_bound} by introducing slack variables $\{ d_{m, n}\}$, $\{ d_n \}$, and $\{ f_n \}$, which satisfy the following bounding conditions:
\begin{flalign}
    & e^{d_{m, n}} \geq \left| \left| {\bm \psi}_m^{\rm Pin} - {\bm \psi}_n^{\rm U} \right| \right|^2, \label{con_d_mn} \\
    & d_n \geq 2 \alpha \left| \left| {\bm \psi}_n^{\rm Pin} - {\bm \psi}_{n, 0}^{\rm Pin}  \right| \right|, \label{con_d_n}  \\
    & e^{f_n}  \geq 2^{R_n} - 1. \label{con_f_n}
\end{flalign}
As a result, constraint \eqref{R_upper_bound} can be reformulated as the following convex form:
\begin{flalign}
    \sigma_n^2  e^{f_n + d_n + d_{n, n}}  \leq \eta  p_n . \label{con_R_up}
\end{flalign}

Next, we turn our attention to the non-convex constraint \eqref{p5b}. The first term in $\Phi_n ( \{ p_m, x_m^{\rm Pin}, R_m \})$ is highly complicated due to the coupling between the summation and multiplication operations. To circumvent this challenge, we establish the following proposition. 
\begin{mypro} \label{proposit1}
For any given variables $\{ \lambda_{m, n} \} > 0$ and $a_n \geq 0$, the following inequality holds for an arbitrary auxiliary variable $\theta_n \in (0, \min_{m \neq n} \lambda_{m, n} )$,
\begin{flalign}
    \sum^N_{m \neq n} e^{ - \lambda_{m, n} a_n} \prod^N_{k \neq m,n} \frac{\lambda_{k, n}}{ \lambda_{k, n} - \lambda_{m, n} } \leq e^{ - a_n \theta_n} \prod^N_{m \neq n} \frac{\lambda_{m, n}}{\lambda_{m,n} - \theta_n}. \label{pro_1_inequality}
\end{flalign}
\end{mypro}
\begin{IEEEproof}
    According to Lemma \ref{hypo_exp_n}, the left-hand side of the inequality \eqref{pro_1_inequality} can be interpreted as the complementary CDF of a hypo-exponential distribution, which yields
    \begin{flalign}
        \sum^N_{m \neq n} e^{- \lambda_{m, n} a_n} \prod^N_{k \neq m, n} \frac{\lambda_{k, n}}{\lambda_{k, n} - \lambda_{m, n}} = \mathbb{P} \left( \sum^N_{m \neq n} z_{m, n} \geq a_n  \right),  \label{pro_1_1}
    \end{flalign}
    where $z_{m, n} \sim {\rm Exp}(\lambda_{m, n})$.

    By applying the Chernoff bounding technique to \eqref{pro_1_1}, we can obtain
     \begin{flalign}
        \mathbb{P} \left( \sum^N_{m \neq n} z_{m, n} \geq  a_n \right) & \leq e^{-a_n \theta_n} \prod^N_{m \neq n} \mathbb{E} \left( e^{ \theta_n z_{m, n} }  \right) \nonumber \\
        & \overset{(a)}{=} 
        e^{-a_n \theta_n} \prod^N_{m \neq n} \frac{\lambda_{m, n}}{ \lambda_{m, n} - \theta_n},
   \end{flalign}
   where $(a)$ follows from the moment generating function of the exponential distribution under the convergence condition $\theta_n < \lambda_{m, n}, ~\forall m \neq n$. This completes the proof.
\end{IEEEproof}

Based on Proposition \ref{proposit1}, we introduce the auxiliary variables $\{ t_n \}$ and slack variables $\{ l_{m, n}, l_n, v_n \}$ to decouple the highly intertwined terms\footnote{To further facilitate the decoupling of variables, the introduced auxiliary variable $\{\theta_n \}$ in Proposition \ref{proposit1} is substituted with $\{ e^{t_n} \}$.}, which satisfy that
\begin{subequations}
\begin{flalign}
    &
    \eta p_m \kappa^2 e^{t_n + f_n} < e^{l_m + l_{m, n}},~ \forall m \neq n, \label{con_t_n} \\
    &
    e^{l_{m, n}} \leq \left| \left| {\bm \psi}_m^{\rm Pin} - {\bm \psi}_n^{\rm U} \right| \right|^2,  \label{con_l_mn} \\
    &
    l_n \leq 2 \alpha \left|  \left| {\bm \psi}_n^{\rm Pin} - {\bm \psi}_{n, 0}^{\rm Pin} \right| \right|, \label{con_l_n} \\
    & 
    v_n + \beta e^{l_{n, n}} + \eta p_n e^{t_n - d_n - d_{n, n}} \geq \sigma_n^2 e^{f_n + t_n}. \label{con_v_n}
\end{flalign}
\end{subequations}
By combining Proposition \ref{proposit1} with the introduced constraints \eqref{con_d_mn}, \eqref{con_d_n}, \eqref{con_f_n}, \eqref{con_t_n}, \eqref{con_l_mn}, \eqref{con_l_n}, and \eqref{con_v_n}, the Chernoff upper bound of the first term in $\Phi_n(\{ p_m, x_m^{\rm Pin}, R_m \})$ can be obtained as
\begin{flalign}
    e^{v_n} \prod_{m \neq n} \left( 1 -  \eta p_m \kappa^2 e^{t_n + f_n - l_m -l_{m, n}} \right)^{-1}. \label{first_up_1_init}
\end{flalign}
To decouple the multiplication operations in \eqref{first_up_1_init}, we introduce slack variables $\{ u_{m, n} \}$ such that
\begin{flalign}
    &
    e^{u_{m, n}} \leq 
    1 -  \eta p_m \kappa^2 e^{t_n + f_n - l_m -l_{m, n}}. \label{con_u_mn}
\end{flalign}
The upper bound in \eqref{first_up_1_init} can be rewritten as
\begin{flalign}
    e^{v_n - \sum_{m \neq n} u_{m, n} }.  \label{first_up}
\end{flalign}

Similarly, by incorporating constraints \eqref{con_d_mn}, \eqref{con_d_n}, \eqref{con_f_n}, \eqref{con_l_mn}, and \eqref{con_l_n}, the second term in $\Phi_n (\{ p_m, x_m^{\rm Pin}, R_m \})$ can be reformulated as
\begin{flalign}
    &
    \left( 1 - e^{ - \beta e^{d_{n, n}} }   \right)
     \Bigg(  1 - e^{- \frac{ \sigma_n^2 e^{f_n + d_n + d_{n, n}} }{\eta p_n \kappa^2 } }  \nonumber \\
     & \times
    \left( \prod\nolimits_{m \neq n} \frac{ 1  }{ 1 +  \frac{p_m }{p_n }  e^{f_n + d_n + d_{n, n} - l_m - l_{m, n} }  }  \right)  \Bigg). \label{second_up_init}
\end{flalign}
To simplify \eqref{second_up_init}, we introduce slack variables $\{ w_{m, n} \}$, $\{ \iota_n \}$, and $\{ \tau_n \}$, which are restricted by
\begin{subequations}
\begin{flalign}
    &
    p_n e^{w_{m, n}} \geq  p_n +  p_m   e^{f_n + d_n + d_{n, n} - l_m - l_{m, n} }, \label{con_w_mn}  \\
    &
     \iota_n  \leq - \beta e^{d_{n, n}}, \label{con_iota_n} \\
     &
     \eta p_n \kappa^2 \tau_n \leq 
     - \sigma_n^2 e^{f_n + d_n + d_{n, n}}. \label{con_tau_n}
\end{flalign}
\end{subequations}
Exploiting the above introduced constraints, \eqref{second_up_init} can be decoupled into the following form
\begin{flalign}
    \left( 1 - e^{ \iota_n }   \right)
     \left(  1 - e^{ \tau_n  -\sum_{m \neq n} w_{m, n} }   \right). \label{second_up}
\end{flalign}

By substituting the derived upper bounds \eqref{first_up} and \eqref{second_up} into constraint \eqref{p5b}, we can obtain that
\begin{flalign}
     &
    e^{v_n - \sum_{m \neq n} u_{m, n} } + e^{ \iota_n + \tau_n  -\sum_{m \neq n} w_{m, n} } + 1 \nonumber \\
    & ~~
     \leq e^{ \iota_n } + e^{ \tau_n  -\sum_{m \neq n} w_{m, n} }  + \epsilon. \label{con_rate_outage}
\end{flalign}

Furthermore, it is observed that $\{p_n\}$ are multiplicatively coupled with other optimization variables in many constraints. To decouple these terms, we execute a variable replacement by defining $p_n = e^{q_n},~\forall n \in \{1, 2, \ldots, N\}$. Therefore, constraints \eqref{p0c}, \eqref{con_R_up}, \eqref{con_t_n}, \eqref{con_v_n}, \eqref{con_u_mn}, \eqref{con_w_mn}, and \eqref{con_tau_n} are respectively re-expressed as
\begin{subequations}
\begin{flalign}
    &
     \sum\nolimits_{n = 1}^N e^{q_n} \leq P_{\max}, 
     \label{con_transmit_q}  \\
    &
      f_n + d_n + d_{n, n} - q_n  \leq \ln \left( \frac{\eta}{\sigma_n^2} \right) ,  \label{con_R_up_q} \\
    & 
    \ln \left( \eta \kappa^2 \right) + q_m + t_n + f_n < l_m + l_{m, n},~ \forall m \neq n, \label{con_t_n_q}  \\
    &
    v_n + \beta e^{l_{n, n}} + \eta e^{q_n + t_n - d_n - d_{n, n}} \geq \sigma_n^2 e^{f_n + t_n}, \label{con_v_n_q}  \\
    &
     e^{u_{m, n}} \leq 
    1 -  \eta \kappa^2 e^{q_m + t_n + f_n - l_m -l_{m, n}}, \label{con_u_mn_q} \\
    &  
    1 \geq  e^{ - w_{m, n}} + e^{q_m + f_n + d_n + d_{n, n} - l_m - l_{m, n} - q_n - w_{m, n} }, \label{con_w_mn_q}  \\
    &
     \eta \kappa^2 \tau_n + \sigma_n^2 e^{f_n + d_n + d_{n, n} - q_n} \leq 0.  \label{con_tau_n_q}
\end{flalign}
\end{subequations}

Through the aforementioned analytical reformulations and tractable approximations, Problem ${\bf P}_5$ can be transformed into the following problem
\begin{subequations}
	\begin{flalign}
		{{\bf P}_6}: & \mathop{\max}  \limits_{ \mathcal{L} } \sum\nolimits_{n=1}^N R_n      \\
		{\rm s.t.} ~ & \eqref{p0d}, \eqref{con_d_mn}, \eqref{con_d_n}, \eqref{con_f_n}, \eqref{con_l_mn}, \eqref{con_l_n}, \eqref{con_iota_n}, \eqref{con_rate_outage},  
         \nonumber \\
        &
        \eqref{con_transmit_q}, \eqref{con_R_up_q}, \eqref{con_t_n_q}, \eqref{con_v_n_q}, \eqref{con_u_mn_q}, \eqref{con_w_mn_q}, \eqref{con_tau_n_q}. \nonumber
	\end{flalign}
\end{subequations}
where  $\mathcal{L} \triangleq$~$\{ x_n^{\rm Pin}, R_n, d_{m, n}, d_n, f_n, t_n, l_{m, n}, l_n, v_n, u_{m, n},$ $w_{m, n}, \iota_n, \tau_n, q_n \}$. 

Problem ${\bf P}_6$ is non-convex due to  constraints \eqref{con_d_mn}, \eqref{con_f_n}, \eqref{con_l_mn}, \eqref{con_l_n}, \eqref{con_rate_outage}, and \eqref{con_v_n_q}. It is observed that both sides of these constraints are convex/concave functions. Consequently, given a feasible point to Problem ${\bf P}_6$, i.e.,  $\widetilde{\mathcal{L}}\triangleq \{ \widetilde{x}_n^{\rm Pin},$$\widetilde{R}_n, \widetilde{d}_{m, n}, \widetilde{d}_n, \widetilde{f}_n, \widetilde{t}_n, \widetilde{l}_{m, n}, \widetilde{l}_n, \widetilde{v}_n,$$\widetilde{u}_{m, n},$$\widetilde{w}_{m, n}, \widetilde{\iota}_n,$$\widetilde{\tau}_n,$$\widetilde{q}_n \}$, these constraints can be respectively approximated by their first-order Taylor expansions as follows
\begin{subequations}
\begin{flalign}
    &
    G \left( d_{m, n}, \widetilde{d}_{m, n} \right) \geq \left| \left| {\bm \psi}_m^{\rm Pin} - {\bm \psi}_n^{\rm U} \right| \right|^2, \label{con_d_mn_tilde}  \\
    &
    G \left( f_n, \widetilde{f}_n \right)  \geq 2^{R_n} - 1, \label{con_f_n_tilde} \\
    & 
    e^{l_{m, n}} \leq \left( \widetilde{x}_m^{\rm Pin} - x_n^{\rm U} \right)^2 + 2 \left( \widetilde{x}_m^{\rm Pin} - x_n^{\rm U} \right) \left( x_m^{\rm Pin} - \widetilde{x}_m^{\rm Pin} \right) \nonumber \\
    &    ~~~~~~~~~
    + \left( y_m^{\rm Pin} - y_n^{\rm U} \right)^2 + H^2, \label{con_l_mn_tilde} \\
    &  
    \frac{l_n}{2 \alpha} \leq  \sqrt{ \left( \widetilde{x}_n^{\rm Pin} - x_{n, 0}^{\rm Pin}  \right)^2 + \left( y_n^{\rm Pin} - y_{n, 0}^{\rm Pin} \right)^2 } \nonumber \\
    & ~~~~~~~~
    +   \frac{ \left( \widetilde{x}_n^{\rm Pin} - x_{n, 0}^{\rm Pin} \right) \left( x_n^{\rm Pin} - \widetilde{x}_n^{\rm Pin}  \right)   }{ \sqrt{\left( \tilde{x}_n^{\rm Pin} - x_{n, 0}^{\rm Pin}  \right)^2 + \left( y_n^{\rm Pin} - y_{n, 0}^{\rm Pin} \right)^2 } },
    \label{con_l_n_tilde}  \\
    &
    e^{v_n - \sum_{m \neq n} u_{m, n} } + e^{ \iota_n + \tau_n  -\sum_{m \neq n} w_{m, n} } + 1 \leq G\left( \iota_n, \widetilde{\iota}_n \right)  \nonumber \\
    &  + G\left( \tau_n  -\sum\nolimits_{m \neq n} w_{m, n}, \widetilde{\tau}_n - \sum\nolimits_{m \neq n} \widetilde{w}_{m, n}  \right)  + \epsilon, \label{con_rate_outage_tilde}  \\
    &
    \sigma_n^2 e^{f_n + t_n} \leq 
    v_n + \beta G \left( l_{n, n}, \widetilde{l}_{n, n} \right) \nonumber  \\
    &
    + \eta G\left( q_n + t_n - d_n - d_{n, n}, \widetilde{q}_n + \widetilde{t}_n - \widetilde{d}_n - \widetilde{d}_{n, n} \right),  \label{con_v_n_tilde} 
\end{flalign}
\end{subequations}
where $G(x, \widetilde{x}) \triangleq  e^{\widetilde{x}}  ( 1 + x - \widetilde{x} ) $.

Then, Problem ${\bf P}_6$ can be approximated by the following convex optimization problem:
\begin{subequations}
	\begin{flalign}
		{{\bf P}_7}: & \mathop{\max}  \limits_{ \mathcal{L} } \sum\nolimits_{n=1}^N R_n      \\
		{\rm s.t.} ~ & \eqref{p0d}, \eqref{con_d_n}, \eqref{con_iota_n}, \eqref{con_transmit_q}, \eqref{con_R_up_q}, \eqref{con_t_n_q}, \eqref{con_u_mn_q}, \eqref{con_w_mn_q},  
         \nonumber \\
        &
        \eqref{con_tau_n_q}, \eqref{con_d_mn_tilde}, \eqref{con_f_n_tilde}, \eqref{con_l_mn_tilde}, \eqref{con_l_n_tilde}, \eqref{con_rate_outage_tilde}, \eqref{con_v_n_tilde}. \nonumber
	\end{flalign}
\end{subequations}
The SCA technique \cite{sca} is employed to enhance the approximation precision, where the local feasible point $\widetilde{\mathcal{L}}$ is updated within each iteration. For clarity, the proposed SCA-based algorithm is summarized in Algorithm \ref{alg:2}.
 \begin{algorithm}
    \caption{The proposed SCA-based algorithm.} \label{alg:2}
\begin{algorithmic}
  \State{\textbf{Initialization}:}
  \State{~ Set $i:=1$ and obtain a feasible $\widetilde{\mathcal{L}}_0$ to Problem ${\bf P}_6$;}
  \While{\rm the stop criterion is not satisfied}
  \State{Obtain the optimal $\mathcal{L}_i^{\star}$ to Problem ${\bf P}_7$ with $\widetilde{\mathcal{L}}_{i-1}$;}
  \State{Update $\widetilde{\mathcal{L}}_i := \mathcal{L}_i^{\star}$;}
  \State{Update $i:=i+1$;}
  \EndWhile
\end{algorithmic}
\end{algorithm}

\section{Simulation Results}
In this section, we comprehensively evaluate the performance of the proposed PGD-based and SCA-based algorithms in both two-user and multi-user scenarios. Unless otherwise specified, the default simulation parameters are configured as follows: the waveguide height is $H = 3~{\rm m}$, the service region length is $D_{\rm L} = 80~{\rm m}$, and the width of each sub-region is $D_{\rm W}/N= 50~{\rm m}$. The in-waveguide attenuation and blockage coefficients are set to $\alpha = 0.0046$ and $\beta=0.01$, respectively, while the additional NLoS signal attenuation factor is $\kappa^2 = -30~{\rm dB}$. The carrier wavelength is $\lambda = 0.01~{\rm m}$, and the AWGN power is $\sigma_n^2 = -120~{\rm dBm}$. The total transmit power budget is $P_{\max} = 10~{\rm mW}$, and the maximum tolerable outage probability requirement is $\epsilon = 0.01$.

To demonstrate the efficacy of our proposed designs, the following three benchmarks are adopted for comparison:
\begin{itemize}
    \item \emph{Exhaustive search}: This scheme implements a grid search over the feasible region with a resolution of $0.1~{\rm m}$ for antenna deployment and $0.1~{\rm mW}$ for power allocation to approach the globally optimal solution. 
    \item \emph{PA-enabled TDMA}: This scheme serves multiple users via TDMA, where all PAs are dynamically steered toward the currently scheduled user, i.e., $x_m^{\rm Pin} = x_n^{\rm U},~\forall m \in \{1, 2, \ldots, N \}$.
    \item \emph{Conventional TDMA}: This scheme serves multiple users via TDMA, where all PAs are statically deployed at the geographical center of the service region, i.e., $x_m^{\rm Pin} = D_{\rm L}/2,~\forall m \in \{1, 2, \ldots, N\}$.
\end{itemize}

\subsection{Two-user Two-PA Scenario}
This subsection considers the two-user two-PA scenario to 
validate the derived closed-form outage probability expressions in \eqref{closed_u1} and \eqref{phi_n}, and to evaluate the performance of the proposed PGD-based algorithm. Furthermore, the capability of the developed SCA-based algorithm for the general multi-user case is also evaluated by specializing it to the case of $N=2$. 

\begin{figure}
	\begin{center}
	\centerline{\includegraphics[ width=.49\textwidth]{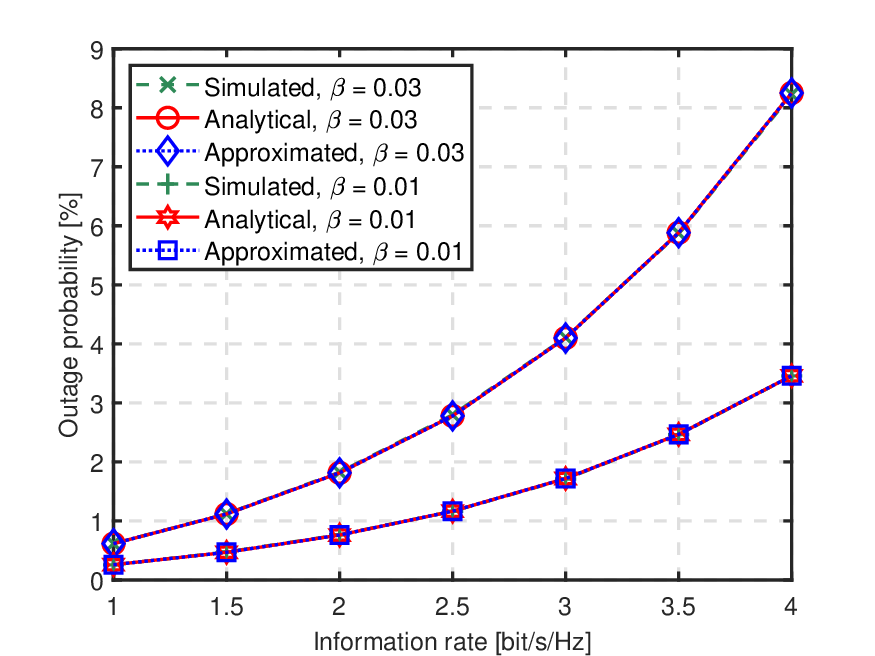}}
    \caption{Comparison of the simulated, analytical, and approximated outage probabilities against the information rate with the blockage coefficient $\beta \in \{0.01, 0.03\}$.}
	\label{fig: prob_sim_ana_appro}
    \vspace{-2em}
	\end{center}
\end{figure}
Fig. \ref{fig: prob_sim_ana_appro} illustrates the outage probability versus the information rate with the blockage coefficient $\beta \in \{0.01, 0.03\}$ under a two-PA case, where the desired PA and interfering PA are deployed $5~{\rm m}$ and $40~{\rm m}$ away from the user with transmit powers of $10~{\rm mW}$ and $5~{\rm mW}$, respectively. The simulated, analytical, and approximated outage probabilities are derived via extensive Monte Carlo realizations, the closed-form expression in equation \eqref{closed_u1}, and the approximation in equation \eqref{phi_n} for the general multi-user case specialized with $N=2$, respectively. One can find that the outage probability increases with the information rate threshold. Moreover, a larger $\beta$ significantly elevates the outage probability, since harsher blockage environments severely disrupt the desired link to be LoS. Notably, there exists a perfect match among the simulated, analytical, and approximated outage probabilities, which not only validates the accuracy of the derived closed-form expression but also demonstrates the tightness of the approximation employed for the general case.

\begin{figure}
	\begin{center}
	\centerline{\includegraphics[ width=.49\textwidth]{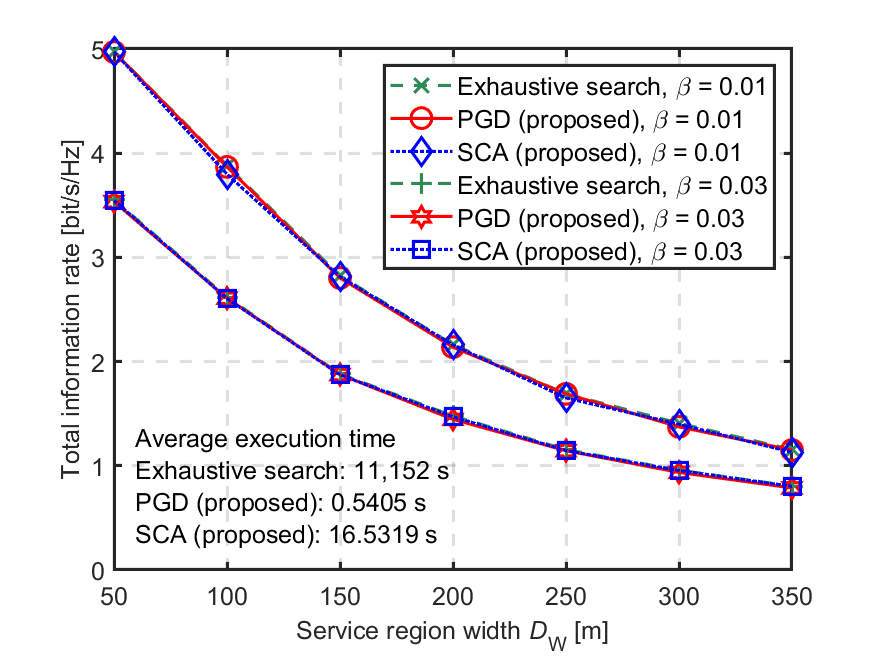}}
    \caption{Comparison of the exhaustive search, PGD-based, and SCA-based algorithms in terms of the total information rate versus the service region width $D_{\rm W}$ with the blockage coefficient $\beta \in \{0.01, 0.03\}$.}
	\label{fig: pgd_sca_exhau}
    \vspace{-2em}
	\end{center}
\end{figure}
Fig. \ref{fig: pgd_sca_exhau} compares the exhaustive search, PGD-based, and SCA-based algorithms in terms of the total information rate versus the service region width $D_{\rm W}$ with the blockage coefficient $\beta \in \{0.01, 0.03\}$. As observed, a larger $\beta$ degrades the total information rate due to the heightened rate outage probability, which aligns with Fig. \ref{fig: prob_sim_ana_appro}. Furthermore, the total information rate obtained by all algorithms  declines as $D_{\rm W}$ expands due to the increased propagation loss and diminished LoS probabilities for the desired links. It is noted that both the proposed PGD-based and SCA-based algorithms achieve near-optimal performance that closely approaches the performance of the exhaustive search, thereby validating their efficacy. Moreover, the proposed PGD-based algorithm yields the lowest average execution time owing to the precise boundary analysis tailored for the two-user two-PA case.

\begin{figure}
	\begin{center}
	\centerline{\includegraphics[ width=.49\textwidth]{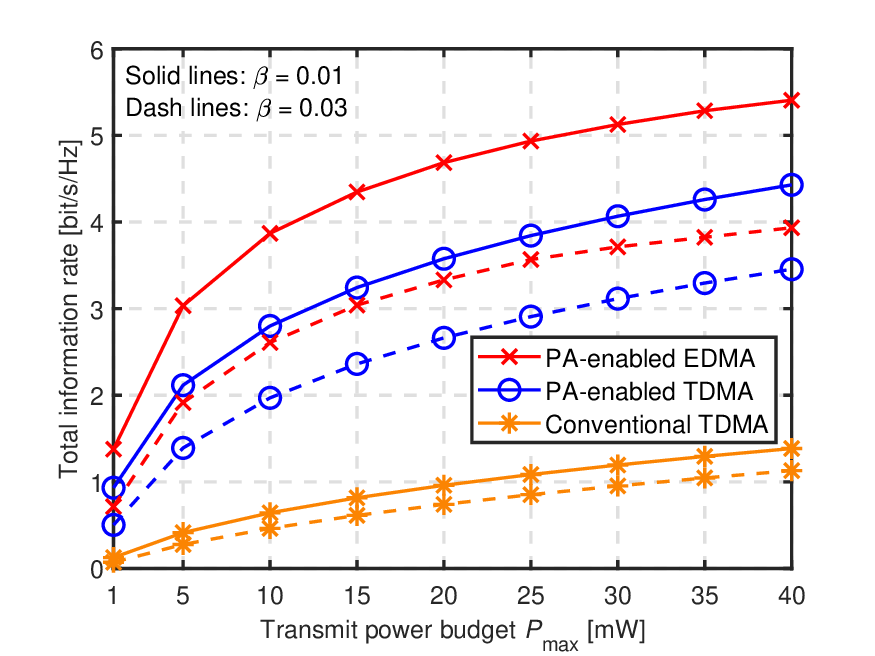}}
    \caption{Comparison of the proposed PA-enabled EDMA design, PA-enabled TDMA design, and conventional TDMA design in terms of the total information rate versus the transmit power budget $P_{\max}$ with the blockage coefficient $\beta \in \{0.01, 0.03\}$.}
	\label{fig: emda_tdma}
    \vspace{-2em}
	\end{center}
\end{figure}
Fig. \ref{fig: emda_tdma} compares the proposed PA-enabled EDMA design, PA-enabled TDMA design, and conventional TDMA design in terms of the total information rate versus the transmit power budget $P_{\max}$ with the blockage coefficient $\beta \in \{0.01, 0.03\}$. It is observed that the total information rate obtained by all designs increases with $P_{\max}$. Moreover, the PA-enabled TDMA design consistently outperforms the conventional TDMA design by dynamically deploying PAs to approach the scheduled user, thereby reducing the path loss and establishing high-quality LoS links. Furthermore, the proposed PA-enabled EDMA design achieves the highest total information rate among all considered designs. The underlying reason is that EDMA can constructively exploit the inherent environmental blockages and user spatial diversity to secure strong LoS paths for desired links, while simultaneously inducing the interfering links into NLoS states. Since the resulting NLoS interfering links are significantly weaker than the LoS desired links, the proposed PA-enabled EDMA design can concurrently serve multiple users in a quasi-interference-free manner.

\subsection{Multi-user Multi-PA Scenario}
In this subsection, we evaluate the capability of the proposed PA-enabled EDMA design under a general multi-user scenario with $N=4$, where the corresponding optimization is implemented by the proposed SCA-based algorithm. 

\begin{figure}
	\begin{center}
	\centerline{\includegraphics[ width=.49\textwidth]{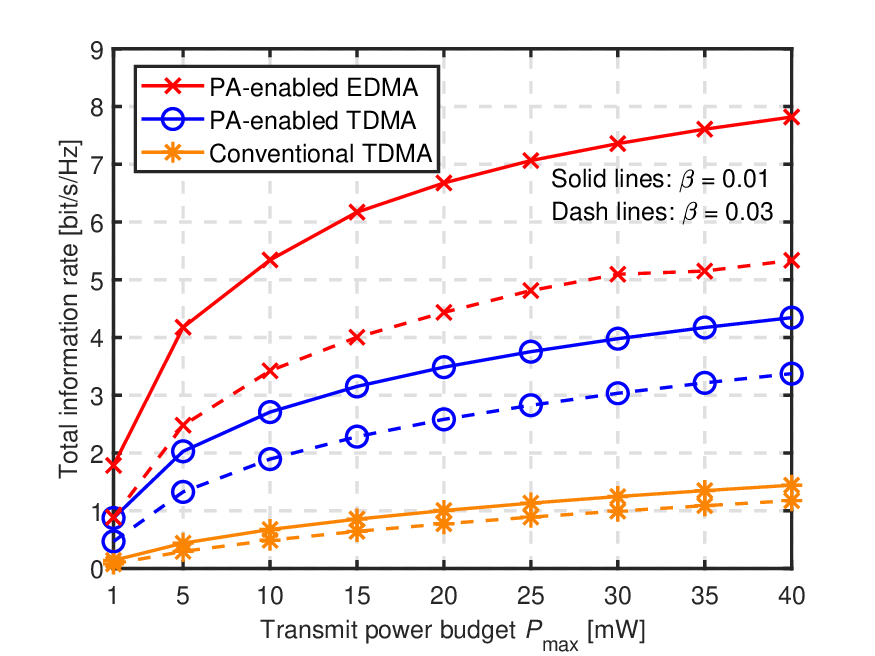}}
    \caption{Impact of the transmit power budget $P_{\max}$ on the PA-enabled EDMA design with the blockage coefficient $\beta \in \{0.01, 0.03\}$.}
	\label{fig: P_max_n}
    \vspace{-2em}
	\end{center}
\end{figure}
Fig. \ref{fig: P_max_n} evaluates the impact of the transmit power budget $P_{\max}$ on the total information rate of the PA-enabled EDMA design with the blockage coefficient $\beta \in \{0.01, 0.03\}$, where the PA-enabled TDMA and conventional TDMA designs are presented for comparison. As observed, the total information rate increases with $P_{\max}$, and the EDMA design achieves the best performance. By comparing Fig. \ref{fig: P_max_n} with Fig. \ref{fig: emda_tdma}, one can find that the performance gain of EDMA over TDMA becomes significantly more pronounced as the number of users increases from $N=2$ to $N=4$. This phenomenon highlights the advantages of EDMA in the multi-user case. Specifically, a larger user cluster inherently offers richer spatial diversity in terms of geographic distribution. EDMA can effectively utilize this heightened spatial diversity to facilitate quasi-interference-free concurrent transmissions by flexibly deploying PAs, while TDMA fails to exploit the spatial distribution of users and suffers from the inherent limitations of time division as the system scales.

\begin{figure}
	\begin{center}
	\centerline{\includegraphics[ width=.49\textwidth]{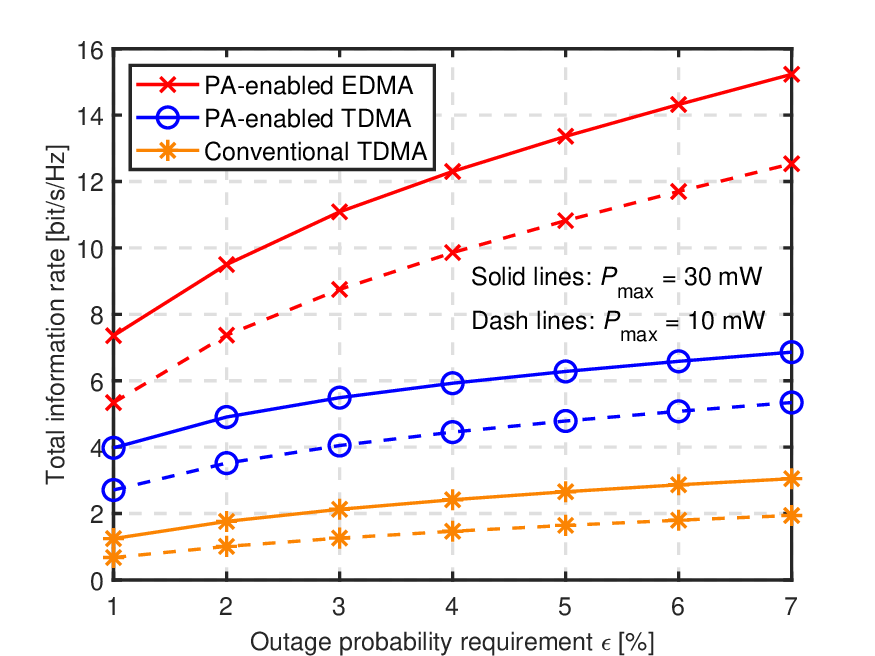}}
    \caption{Impact of the outage probability requirement $\epsilon$ on the PA-enabled EDMA design with the transmit power budget $P_{\max} \in \{10~{\rm mW}, 30~{\rm mW} \}$.}
	\label{fig: eps_n}
    \vspace{-2em}
	\end{center}
\end{figure}
Fig. \ref{fig: eps_n} evaluates the impact of the outage probability requirement $\epsilon$ on the total information rate of the PA-enabled EDMA design with the transmit power budget $P_{\max} \in \{10~{\rm mW}, 30~{\rm mW}\}$, where PA-enabled TDMA and conventional TDMA designs are presented as benchmarks. It is observed that the total information rate obtained by all designs increases with $\epsilon$ as the outage constraints are systematically relaxed. Moreover, the proposed EDMA design outperforms the other two benchmarks, and the performance gap expands as $\epsilon$ grows. The reason is that a larger $\epsilon$ broadens the optimization feasible region. EDMA can fully exploit this expanded design space to obtain synergistic gains via the joint optimization of PA positioning and power allocation adaptive to the spatial distribution of users, while such multidimensional flexibility is inherently missing in TDMA designs.

\begin{figure}
	\begin{center}
	\centerline{\includegraphics[ width=.49\textwidth]{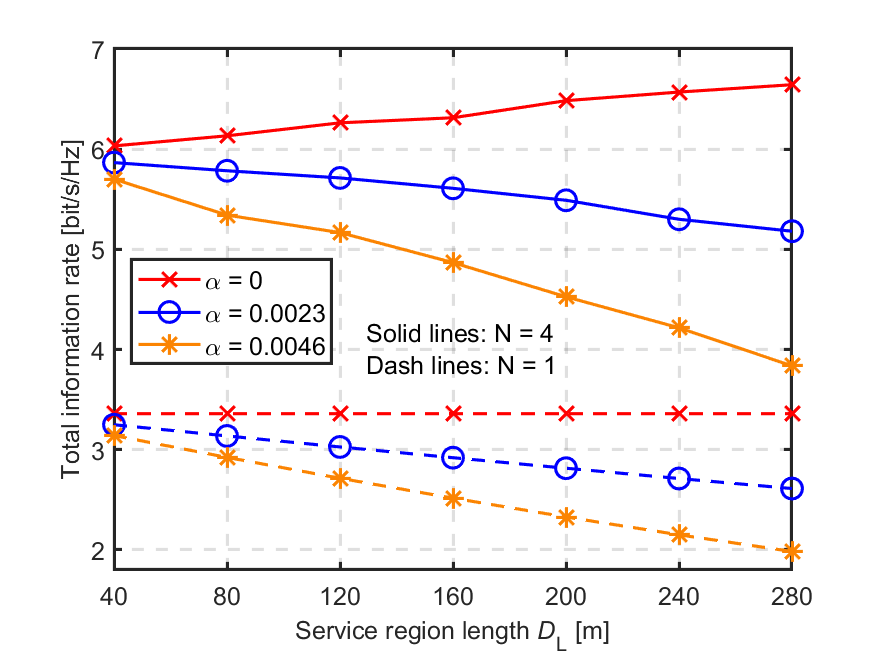}}
    \caption{Impact of the service region length $D_{\rm L}$ on the PA-enabled EDMA design with the in-waveguide attenuation coefficient $\alpha \in \{0, 0.0023, 0.0046\}$ and the number of users $N \in \{1, 4\}$.}
	\label{fig: DL_n}
    \vspace{-2em}
	\end{center}
\end{figure}
Fig. \ref{fig: DL_n} evaluates the impact of the service region length $D_{\rm L}$ on the total information rate of the PA-enabled EDMA design with the in-waveguide attenuation coefficient $\alpha \in \{0, 0.0023, 0.0046\}$ and the number of users $N \in \{1, 4\}$. One can find that for the single-user lossless scenario ($\alpha = 0$ and $N=1$), the total information rate remains invariant with $D_{\rm L}$. The reason is that the absence of in-waveguide attenuation enables PAs to flexibly adapt their lateral positions within the waveguide to losslessly track the user. Furthermore, in the multi-user lossless scenario ($\alpha = 0$ and $N=4$), the total information rate increases with $D_{\rm L}$ since a larger service region can provide richer diversity of user distribution, which can be employed by EDMA to mitigate inter-user interference. On the contrary, when the in-waveguide attenuation is non-negligible ($\alpha \in \{0.0023, 0.0046\}$), the total information rate decreases with $D_{\rm L}$, as the signal degradation caused by in-waveguide attenuation outweighs the spatial multi-user isolation gains offered by the expanded deployment region.

\section{Conclusion}
This paper investigated the joint design of PA deployment and power allocation in PA-enabled outage-constrained EDMA systems. To capture practical propagation environment characteristics, probabilistic LoS blockages, random NLoS scattering, and in-waveguide attenuation were considered in the model formulation. We maximized the total transmission rate under the constraints of statistical rate outage, transmit power budget, and feasible deployment region of PAs. First, we considered a special case involving two users and two PAs, and derived the corresponding closed-form expressions for outage probabilities, followed by a PGD-based algorithm to handle the reformulated problem. Then, we extended our design to the general multi-user multi-PA scenarios and derived the tractable approximations for outage probabilities, where the resulting problem is addressed by a proposed SCA-based algorithm. Simulations verified that both PGD-based and SCA-based algorithms achieved the near-optimal performance in comparison with exhaustive search. Moreover, there existed significant performance gains of the proposed EDMA design over TDMA designs.



\begin{thebibliography}{1}
\bibitem{bc_6g}
C. -X. Wang et al., ``On the road to 6G: Visions, requirements, key technologies, and testbeds," \emph{IEEE Commun. Surveys Tuts.}, vol. 25, no. 2, pp. 905-974, 2nd Quart., 2023.

\bibitem{bc_spectrum}
B. He et al., ``Toward secure semantic transmission in the era of GenAI: A diffusion-based framework," \emph{IEEE Commun. Mag.}, early access, 2025.

\bibitem{bc_flex}
W. Xiong et al., ``Secure analog beamforming for multi-user MISO systems with movable antennas," \emph{IEEE Trans. Wireless Commun.}, vol. 25, pp. 8164-8178, 2026.

\bibitem{bc_fa}
K.-K. Wong, A. Shojaeifard, K.-F. Tong, and Y. Zhang, ``Fluid antenna systems,” \emph{IEEE Trans. Wireless Commun.}, vol. 20, no. 3, pp. 1950–1962,
Mar. 2021.

\bibitem{bc_ma}
L. Zhu, W. Ma, and R. Zhang, ``Modeling and performance analysis for movable antenna enabled wireless communications,” \emph{IEEE Trans. Wireless Commun.}, vol. 23, no. 6, pp. 6234–6250, Jun. 2024.

\bibitem{bc_ma2}
C. Zhou et al., ``Frequency-switching array enhanced physical-layer security in terahertz bands: A movable antenna perspective," \emph{arXiv preprint arXiv:2507.01624}, 2025.


\bibitem{bc_pa1}
Y. Xu, D. Xu, X. Yu, S. Song, Z. Ding, and R. Schober, ``Joint radiation power, antenna position, and beamforming optimization for pinching-antenna systems with motion power consumption," \emph{IEEE Trans. Wireless Commun.}, vol. 25, pp. 7825-7841, 2026.

\bibitem{bc_pa2}
Z. Ding, R. Schober, and H. V. Poor, ``Flexible-antenna systems: A pinching-antenna perspective,” \emph{IEEE Trans. Commun.}, vol. 73, no. 10, pp. 9236–9253, Oct. 2025.

\bibitem{bc_pa3}
Y. Zhao et al., ``Reconfigurable antennas for next-generation mobile communication networks: A comprehensive survey and tutorial," \emph{IEEE Commun. Surveys Tuts.}, vol. 28, pp. 5267-5306, 2026.


\bibitem{bc_pa_pl}
X. Fan and L. Peng, ``Toward flexible and intelligent indoor NFC: The role of pinching antennas in 6G networks," \emph{IEEE Netw.}, vol. 39, no. 6, pp. 70-77, Nov. 2025.

\bibitem{bc_pa_los}
Y. Zhong, Y. Xiao, Y. Li, H. Chen, X. Lei, and P. Fan, ``2-D pinching-antenna systems: Modeling and beamforming design," \emph{IEEE Trans. Commun.}, vol. 74, pp. 8254-8266, 2026.


\bibitem{rw_pa_sing_w1}
C. Ouyang, Z. Wang, Y. Liu, H. Shin, and Z. Ding, ``Capacity characterization of pinching-antenna systems," \emph{IEEE Trans. Wireless Commun.}, vol. 25, pp. 10387-10404, 2026.

\bibitem{rw_pa_sing_w2}
Y. Xu, Z. Ding, D. Cai, and V. W. S. Wong, ``QoS-aware NOMA design for downlink pinching-antenna systems," \emph{IEEE Trans. Commun.}, vol. 73, no. 12, pp. 13611-13625, Dec. 2025.

\bibitem{rw_pa_mul_w1}
J. Zhao, X. Mu, K. Cai, Y. Zhu, and Y. Liu, ``Waveguide division multiple access for pinching-antenna systems (PASS)," \emph{IEEE Trans. Wireless Commun.}, vol. 25, pp. 13761-13775, 2026.

\bibitem{rw_pa_mul_w2}
K. Wang, Z. Ding, and G. K. Karagiannidis, ``Antenna activation and resource allocation in multi-waveguide pinching-antenna systems," \emph{IEEE Trans. Wireless Commun.}, vol. 25, pp. 4070-4082, 2026.

\bibitem{rw_pa_mul_f1}
W. Mao, Y. Lu, Y. Xu, B. Ai, O. A. Dobre, and D. Niyato, ``Multi-waveguide pinching antennas for ISAC," \emph{IEEE Trans. Wireless Commun.}, vol. 25, pp. 5846-5858, 2026.

\bibitem{rw_pa_mul_f2}
Y. Ai, X. Mu, P. Si, and Y. Liu, ``Delay minimization in pinching-antenna-enabled NOMA-MEC networks," \emph{IEEE Commun. Lett.}, vol. 30, pp. 962-966, 2026.

\bibitem{bc_pro_pa1}
Y. Xu, Y. Lu, Z. Ding, and T. -H. Chang, ``Pinching-antenna system design under random LoS and NLoS channels," \emph{IEEE Trans. Wireless Commun.}, vol. 25, pp. 18440-18455, 2026.

\bibitem{bc_pro_los1}
X. Lan, X. Tang, R. Zhang, Q. Du, and T. Q. S. Quek, ``Time-slotted multi-cluster UAV airComp with energy-awareness: A pointer network-assisted soft actor-critic learning framework," \emph{IEEE J. Sel. Topics Signal Process.}, early access, 2026.

\bibitem{bc_pro_los2}
X. Lan et al., ``UAV-assisted integrated communication and over-the-air computation with interference awareness," \emph{IEEE Trans. Commun.}, vol. 73, no. 11, pp. 10647-10661, Nov. 2025.

\bibitem{bc_pro_pa2}
Y. Xu et al., ``Generalized pinching-antenna systems: A tutorial on principles, design strategies, and future directions," \emph{IEEE Commun. Surveys Tuts.}, vol. 28, pp. 5872-5908, 2026.

\bibitem{bc_pro_pa3}
Y. Xu et al., ``Environment-aware network-level design of generalized pinching-antenna systems—part I: Traffic-aware case," \emph{arXiv preprint arXiv:2602.17023}, 2026.

\bibitem{rw_deter_pa1}
X. Xie, F. Fang, Z. Ding, and X. Wang, ``Pinching antennas in blockage-aware environments: Modeling, design, and optimization," \emph{arXiv preprint arXiv:2601.01277}, 2026.

\bibitem{rw_deter_pa2}
Y. Xu et al., ``Environment-aware network-level design of generalized pinching-antenna systems–part II: Geometry-aware case," \emph{arXiv preprint arXiv:2602.17032}, 2026.

\bibitem{los_block}
Z. Ding, R. Schober, and H. V. Poor, ``Environment division multiple access (EDMA): A feasibility study via pinching antennas,"  \emph{IEEE Trans. Wireless Commun.}, vol. 25, pp. 15675-15691, 2026.



\bibitem{neff}
D. M. Pozar, \emph{Microwave Engineering}, 4th ed. Wiley, New York, US, 1998.

\bibitem{in_wave_att}
J. F. Bauters, M. J. Heck, D. D. John, J. S. Barton, C. M. Bruinink, A. Leinse, R. G. Heideman, D. J. Blumenthal, and J. E. Bowers, ``Planar waveguides with less than 0.1 dB/m propagation loss fabricated with wafer bonding,” \emph{Opt. Express}, vol. 19, no. 24, pp. 24 090–24 101, 2011.

\bibitem{eta_para}
C. Zhou, C. You, S. Shi, J. Zhou, and C. Wu, ``Super-resolution wideband beam training for near-field communications with ultralow overhead," \emph{IEEE Internet Things J.}, vol. 12, no. 14, pp. 28793-28808, Jul. 2025.


\bibitem{p_LoS_channel}
C. You and R. Zhang, ``Hybrid offline-online design for UAV-enabled data harvesting in probabilistic LoS channels," \emph{IEEE Trans. Wireless Commun.}, vol. 19, no. 6, pp. 3753-3768, Jun. 2020.

\bibitem{p_LoS_channel2}
3rd Generation Partnership Project (3GPP), ``Study on channel model for frequencies from 0.5 to 100 GHz (Release 18),” Dec. 2024.

\bibitem{implicit}
S. G. Krantz and H. R. Parks, \emph{The Implicit Function Theorem: History, Theory, and Applications}. Boston, MA, USA: Birkhauser, 2002.


\bibitem{hypo-exp}
S. M. Ross, \emph{Introduction to Probability Models}, 9th ed. Amsterdam, The Netherlands: Academic Press, 2007.

\bibitem{sca}
W. Mao et al., ``Cramér–Rao bound optimization for bistatic ISAC: Transceiver design and attention-based ISACNet," \emph{IEEE J. Sel. Areas Commun.}, vol. 44, pp. 181-195, 2026.



\end{thebibliography}
\end{document}